\documentclass[10pt,conference,letterpaper,nofonttune]{IEEEtran}
\usepackage[utf8]{inputenc}
\usepackage[T1]{fontenc}
\usepackage{graphicx}
\usepackage{enumitem}
\usepackage{longtable}
\usepackage{float}
\usepackage{wrapfig}
\usepackage{rotating}
\usepackage[normalem]{ulem}
\usepackage[fleqn]{amsmath}
\usepackage{textcomp}
\usepackage{marvosym}
\usepackage{wasysym}
\usepackage{amssymb}
\usepackage{hyperref}
\usepackage{semantic}
\usepackage{mathpartir}
\usepackage{amsthm}
\usepackage{hyperref}
\usepackage{epstopdf}
\usepackage{listings}
\usepackage[dvipsnames]{xcolor}
\usepackage{mdframed}
\usepackage{ stmaryrd }
\usepackage{microtype}
\usepackage{times} % this saves space ; 2018-07-10; AA
\usepackage{zlmtt}
\usepackage{csquotes}
\usepackage{enumitem}
\setlist[description]{leftmargin=\parindent,labelindent=\parindent}
\usepackage[clock]{ifsym}
\usepackage[style=ieee, doi=false,isbn=false,url=false, sorting=nty, mincitenames=1, maxcitenames=2]{biblatex}
\AtBeginBibliography{\footnotesize}
\usepackage[extended]{optional}
\newcommand{\basicCodeStyle}{\ttfamily\small}

\newcommand{\keywordCodeStyle}{\bfseries}
\newcommand{\code}[1]{\ensuremath{\mathsf{#1}}}

\lstdefinelanguage{troupe}{
keywords={@type, andalso, fun, let, in, case, of, class, type,
rec, not, as, val, if, then, else, true, false, end, raisedTo, tini, attenuate},
sensitive=true,
commentstyle=\small\itshape\ttfamily\textcolor{gray},
keywordstyle=\ttfamily\underbar,
identifierstyle=\ttfamily,
basewidth={0.5em,0.5em},
columns=fixed,
fontadjust=true,
literate={=>}{{$\Rightarrow$}}3 {===}{{$\equiv$}}1 {=/=}{{$\not\equiv$}}1 {|>}{{$\triangleright$}}3 {\\/}{{$\vee$}}2 {/\\}{{$\wedge$}}2 {^}{{$\uparrow$}}1,
morecomment=[s]{(*}{*)}
}

\newcommand\codebraces[1]{ \code{\{}  #1 \code{\}}}

\newcommand\commandmeta{c}

\newcommand\val{v}

\newcommand\cmd\commandmeta

\newcommand\expr{e}

\newcommand\mem{m}

\newcommand\parenone[1]{ ( #1 ) }

\newcommand\lookupone[2]{#1 \parenone{#2}}
\newcommand\lookupEnv\lookupone

\newcommand\Assign[2]{#1 = #2}
\newcommand\IfKw{\code{if}}
\newcommand\ThenKw  {\code{then}}
\newcommand\ElseKw  {\code{else}}

\newcommand{\IfThenElse}[3]{\IfKw\ #1\ \ThenKw\ {#2} \ \ElseKw\ {#3 }}

\newcommand{\Stop}{ \mathbf{Stop} }

\newcommand\block\codebraces

\newcommand\binop{\mathrel{\mathit{op}}}

\newcommand{\pcdeclfree}[1]{\mathsf{pc\_decl\_free}(#1)}
\newcommand{\declWF}[1]{\mathsf{WF}(#1)}

\newcommand\steps{\longrightarrow}
\newcommand\stepsto\steps %AN ALIAS
\NewDocumentCommand{\stepsmany}{o}{
\IfNoValueTF{#1}{\steps^\ast}{\steps^{#1}}
}

\NewDocumentCommand{\stepswith}{sm}{
\IfBooleanTF{#1}
{\xrightarrow{#2}_{\level_\adv}}
{\stepsto_{#2}}
}
\NewDocumentCommand{\stepsmanywith}{sO{*}m}{
\IfBooleanTF{#1}
{\xrightarrow{#3}^{#2}_{\level_\adv}}
{\stepsto^{#2}_{#3}}
}

\newcommand\lbridgestep[3]{\mathrel{\curvearrowright^{#1,#2}_{#3}}}
\newcommand{\obs}[2]{\flowsTo{#1}{#2}}
\newcommand{\nobs}[2]{\notFlowsTo{#1}{#2}}

\newcommand\flowsTo[2]{#1 \sqsubseteq #2}
\newcommand\notFlowsTo[2]{#1 \not \sqsubseteq #2}

\newcommand\configTwo[3]{\langle #1, #2  \rangle}
\newcommand\configThree[3]{\langle #1, #2, #3 \rangle}
\newcommand\configFour[4]{\langle #1, #2, #3, #4 \rangle}

\newcommand\configThreeWithAux[4]{\langle #1, #2, #3 \mid #4 \rangle}

\newcommand\configFiveWithAux[6]{\langle #1, #2, #3, #4, #5 \mid #6 \rangle}
\newcommand\configSixWithAux[7]{\langle #1, #2, #3, #4, #5, #6 \mid #7 \rangle}
\newcommand\configSevenWithAux[8]{\langle #1, #2, #3, #4, #5, #6, #7 \mid #8 \rangle}
\newcommand\configTwoWithAux[3]{\langle #1, #2, \mid #3 \rangle}

\newcommand\confInst\configFiveWithAux
\newcommand\confTerm\configSixWithAux
\newcommand\conf\configSevenWithAux
\newcommand\configCmd\configTwoWithAux
\newcommand\configExpr\configFour
\newcommand\configClock\configThreeWithAux
\newcommand\tsEvent[2]{(#1,#2)}

\newtheorem{definition}{Definition}
\newtheorem{lemma}{Lemma}
\newtheorem{theorem}{Theorem}

\newcommand\defn{\triangleq}
\newcommand\dom[1]{\mathsf{dom} \parenone{#1}}

\newcommand\heaporder\preceq

\newcommand\pc{\mathit{pc}}
\newcommand\ts{\mathit{ts}}

\newcommand\adv{\mathit{adv}}

\newcommand\eventMeta{\alpha}

\newcommand\metaList\overline

\newcommand\cfgmeta{\mathit{cfg}}
\newcommand\indmeta{I}

\newcommand\indprop[5]{\mathrel{ \left[ #1 \right]^{#2,#3}_{#4 \mid #5} }}

\newcommand\indnarr[5]{\mathrel{ \left\langle #1 \right\rangle^{#2,#3}_{#4 \mid #5} }}
\newcommand\syncbridge[2]{\mathrel{\Rightarrow^{#1}_{#2}}}
\newcommand\syncConfig[4]{\configThree{#1}{#2 \mid #3}{#4}}
\newcommand\emptylist{\mathsf{nil}}

\newcommand*{\fullref}[1]{\hyperref[{#1}]{\autoref*{#1} (\nameref*{#1})}}

\newcommand\Skip{\code{skip}}
\newcommand\While[2]{\code{while}\ #1\ \code{do}\ #2  }
\newcommand\Tini[4]{\code{tini}_{#1}\ \code{to}\ #2\ \code{with}\ #3\ \code{do}\ #4  }

\newcommand\Declassify[4]{\Assign{#1}{\code{decl}\ #4\ \code{to}\ #3\ \code{with}\ #2}}

\newcommand\Eval[2]{\code{eval}\ #1\ #2 }

\newcommand\Attenuate[2]{\code{attenuate}\ #1\ \code{to}\ #2}

\newcommand\AuthorityVal[2]{\mathsf{auth}\ #1\  #2 }

\renewcommand\configTwo[2]{\langle #1, #2  \rangle}

\renewcommand\Stop{\code{stop}}
\newcommand\PCDecl[3]{\code{pcdecl}_{#1} ( #2, #3) }

\newcommand\EventAssign[2]{a( #1, #2 )}

\newcommand{\EventTiniExitSym}{\bar{t}}

\newcommand{\EventTiniExit}[4]{\EventTiniExitSym_{#1}( #2, #4)}

\newcommand\EventEmpty{\epsilon}

\newcommand{\EventDeclassifySym}{d}
\newcommand{\EventDeclassify}[4]{\EventDeclassifySym(#1, #2, #4)}

\newcommand\unused\gamma

\newcommand\level\ell
\newcommand\authpurposemeta{\mathit{p}}

\newcommand\basemeta{\mathit{base}}

\newcommand\LabeledValue[2]{\langle #1; #2 \rangle }

\newcommand\TheLattice{\mathcal{L}}

\newcommand\RootAuth{\code{rootauth}}

\newcommand\flowsto{\sqsubseteq}

\newcommand\yields{\Downarrow}

\renewcommand\stepsto\longrightarrow

\newcommand\levelof[1]{\mathit{lev} (#1)}
\newcommand\typeof[1]{\mathit{type} (#1)}
\newcommand\parseCmd[1]{\mathit{parse} (#1)}

\newcommand\Vars[1]{\mathit{vars} (#1)}
\newcommand\evalFree[1]{\mathit{eval\textnormal{-}free} (#1)}
\newcommand\knowledge{k}
\newcommand\progressknowledge{k_{\rightarrow}}
\newcommand\clockknowledge{k_{\rightarrow}^{\VarClock}}

\newcommand\events{t}
\newcommand\ladv{\mathit{adv}}

\newcommand\loweq[1]{\mathrel{\sim_{#1}}}

\newcommand\tproj[2]{\lfloor #1 \rfloor_{#2}}

\title{Reconciling progress-insensitive noninterference and declassification
}

 \author{
   \IEEEauthorblockN{Johan Bay}
   \IEEEauthorblockA{Aarhus University \\ bay@cs.au.dk}
   \and
   \IEEEauthorblockN{Aslan Askarov}
   \IEEEauthorblockA{Aarhus University \\ aslan@cs.au.dk}
 }

\makeatother
\opt{extended}{\IEEEoverridecommandlockouts
\IEEEpubid{\makebox[\columnwidth]{000-0-0000-0/00\$00.00~\copyright~2020 IEEE \hfill} \hspace{\columnsep}\makebox[\columnwidth]{ }}}

\begin{document}

\maketitle
\IEEEpubidadjcol
\IEEEpeerreviewmaketitle

\begin{abstract}
  Practitioners of secure information flow often face a design challenge: what is the right semantic treatment of leaks via termination? On the one hand, the potential harm of untrusted code calls for strong progress-sensitive security. On the other hand, when the code is trusted to not aggressively exploit termination channels, practical concerns, such as permissiveness of the enforcement, make a case for settling for weaker, progress-insensitive security. This binary situation, however, provides no suitable middle point for systems that mix trusted and untrusted code. This paper connects the two extremes by reframing progress-insensitivity as a particular form of declassification.  Our novel semantic condition reconciles progress-insensitive security as a declassification bound on the so-called progress knowledge in an otherwise progress or timing sensitive setting. We show how the new condition can be soundly enforced using a mostly standard information-flow monitor. 
  We believe that the connection established in this work will enable other applications of ideas from the literature on declassification to progress-insensitivity.
\end{abstract}

\section{Introduction}
\label{sec:introduction}
Progress-insensitive noninterference (PINI) is a popular semantic condition for secure information flow. PINI generalizes the classical termination-insensitive noninterference to accommodate I/O interactions and provides a practical foundation for many information flow systems. A known downside of PINI is that it permits leaking arbitrary amounts of information~\cite{tini}. Malicious code may launder data through termination channels by unary encoding the information in the length of the trace or via timing channels. For these reasons, the consensus in the information flow community is to use PINI for trusted settings, where the goal is to prevent accidental information leaks. For untrusted settings, stronger notions of security, such as progress or timing sensitivity, are necessary.

Many practical scenarios, however, combine both trusted and untrusted code. Such combinations are natural to browser mashups, mobile apps, and just about any system that embeds third-party code. The binary consensus provides no suitable middle ground here. Progress-insensitivity is too permissive, whereas progress and timing-sensitivity is too restrictive.

Consider one such example scenario of a mashup that embeds a third-party newsfeed widget. The widget downloads the latest newsfeed from the news server and displays the favorite topic of the user. The choice of the favorite topic is sensitive and, therefore, must not leak to the news server.
Figure \ref{fig:news-widget} presents a pseudo-code for such a widget.
The widget implements a custom caching logic by maintaining a counter and re-fetching the news on every tenth invocation.
For the purpose of this example, we regard the counter as sensitive as well.
 
The code in Figure~\ref{fig:news-widget} is straightforward and  unproblematic. We can imagine crafting a tool that analyzes (statically or dynamically) the code in Figure~\ref{fig:news-widget} for potential information flow violations. But if we are to take the next step and try to prove our tool sound, we hit a semantic conundrum. Because Line~\ref{fig:widget:receive} contains a potentially blocking network operation, it is unclear how long it may take for the server to respond, if ever. This means that if we want our tool to accept programs such as Figure~\ref{fig:news-widget}, we cannot use progress and timing-sensitive security as the basis for soundness. With the binary consensus, the only other option is progress-insensitive security. This option permits blocking and divergence, making it suitable for Figure~\ref{fig:news-widget}. However, it also forces us to place the termination and timing attacks outside of the formal threat model, which weakens our tool.

\begin{figure}
\begin{center}
\begin{lstlisting}[language=C]
function newsWidget (userFavTopic) {
  if (counter % 10 == 0) {
    feed = receive (newsfeed_server_url) /*@ \label{fig:widget:receive} @*/
  }
  counter ++;
  newstext = feed[userFavTopic]
}  
\end{lstlisting}
\end{center}
\caption{Newsfeed widget code\label{fig:news-widget}}
\end{figure}

This paper addresses the problem of the binary situation by presenting a novel semantic definition that connects the two extremes by reconciling progress-insensitive security as a particular form of declassification. 
This reframing means that we can treat progress-insensitivity just like any other declassification -- a selective weakening of a baseline end-to-end security policy.
It also means that we can transfer insights about declassification policies, such as their dimensions and principles~\cite{Dimensions}, to progress-insensitivity.
The key to the new definition is the use of the epistemic approach to information flow, which allows us to specify a bound on the knowledge the attacker learns from observing the progress of the computation in an otherwise progress or timing-sensitive setting. 

Two meta-level points about our definition are worth highlighting.
First, we note that the practice of declassifying termination leaks by itself is not novel. This idea appears in the literature as early as two decades ago in Jif~\cite{Jif} in the context of programming languages and later in HiStar~\cite{Histar} in the context of operating systems. Here, our work provides a firm theoretical basis that this practice lacked. 
In fact, we show that a mostly standard flow-insensitive 
dynamic monitor soundly enforces the new definition.

Second, we stress the value of the epistemic approach in formulating a concise and intuitive definition. It is not clear to us whether the definition can be reformulated in a classical two-trace style while retaining the same degree of clarity. The discussion of the soundness of our monitor presents an operational security invariant that does have the classical two-trace formulation, but that invariant is far from intuitive.

We present our condition in the setting of a simple imperative language with a standard flow-insensitive dynamic monitor, which conveys the condition in a clean form. The simple language does not contain networking or blocking primitives. This omission does not remove generality from our setup because the language already contains the possibility of divergence via infinite loops. 
We have implemented the enforcement of this condition in Troupe~\cite{Troupe} -- a research programming language with dynamic information flow control, actor-based concurrency, and primitives for distributed programming.

The rest of the paper is structured as follows:
Section~\ref{sec:language} introduces the formal setting of a small imperative language we use in this work. 
The presentation of the security condition is split across two sections. Section~\ref{sec:security-condition} presents the security for a progress-sensitive attacker and presents how a mostly standard dynamic monitor can soundly enforce this condition; Section~\ref{sec:timingsensitivity} presents the security condition for a timing-sensitive attacker. 
We discuss our definitions in Section~\ref{sec:future-work} and report on the implementation
experience in Section~\ref{sec:implementation}. Finally, in Sections \ref{sec:related-work} and \ref{sec:conclusion} we discuss related work and conclude.

\section{The security model and the language}
\label{sec:language}
\subsection{Security model}
\label{sec:lang:secmodel}
We assume a standard security lattice $\TheLattice$ of security
levels~$\level$, with distinguished bottom and top levels
$\bot$ and $\top$, and the operations for least upper bound
$\sqcup$ and the lattice order~$\sqsubseteq$.

Our language is a standard imperative language
extended with capability-based declassification,
and a special purpose \code{tini} command for bounded progress-insensitivity that we explain below.
Each variable in the program has a fixed security level $\levelof{x}$ that does not change throughout the execution.
An attacker associated with a security level $\ell$ observes updates to variables with levels up to $\ell$; they
additionally observe the reachability of the \code{tini} blocks, as we explain below.

In the examples we show here, we use a two or three-level lattice with levels $\mathit{L},\mathit{M},\mathit{H}$, where
$\mathit{L} \flowsto \mathit{M} \flowsto \mathit{H}$, and $\ell \flowsto \ell$ for each $\ell \in \{\mathit{L},\mathit{M},\mathit{H}\}$.
We adopt the convention of using upper-case letters to denote concrete lattice elements of $\TheLattice$
and lower-case letters to denote variables of said level. 
As such, $h_1$ and $h_2$ are variables such that $\levelof{h_1} = \levelof{h_2} = H$.

\subsection{The language and the monitoring semantics}
Figure~\ref{fig:syntax} presents the syntax of our language. We explain the formal semantics of the language and then discuss the non-standard features.
\begin{figure}
\fbox{
\begin{minipage}{0.455\textwidth}
\begin{align*}
  \expr ::= &\ n \mid x \mid \expr \binop \expr \mid \Attenuate {\expr} { (\level, p) }  \\
  c ::= &\ \Skip \mid c; c \mid \While{\expr}{\cmd} \mid \IfThenElse{\expr}{\cmd}{\cmd}
  \\ & \mid \Assign{x}{\expr} \mid \Tini{\eta}{\level}{e}{\cmd}
  \\ & \mid \Declassify{x}{\expr}{\level}{\expr} 
  \\ & \mid \Eval \expr \{x_1, \ldots, x_n\} 
\end{align*}    
\end{minipage}
}

\caption{Syntax of the language\label{fig:syntax}}
\end{figure}

\paragraph{Monitoring semantics}
For evaluating commands we use a small-step semantics transition
$\configThree{\cmd}{\mem}{\pc} \stepswith{\eventMeta} \configThree{\cmd'}{\mem'}{\pc'}$,
where $\pc$ is the security level of the program counter, and $\eventMeta$ is the event generated by the step.
The events can be empty events, denoted by $\EventEmpty$, and assignments and
declassifications per the following grammar:
$$\eventMeta ::= \EventEmpty \mid
\EventAssign{x}{\val} \mid
\EventDeclassify{x}{\level}{\level}{\level} \mid
\EventTiniExit{\eta}{\level}{\level}{\level}
$$
The $\Stop$ and \code{pcdecl} commands are only used internally, and therefore not part of the syntax of the language.
Command~$\Stop$ denotes final configurations that cannot step any further.
For evaluating expressions we use a big-step relation $\configTwo{e}{m} \yields \LabeledValue{\basemeta}{\level}$ that relates an expression with a labeled value.
Labeled values $\LabeledValue{\basemeta}{\level}$ consists of a base value and a level, where $\level$ denotes the confidentiality-level of the base value $\basemeta$.
Base values include integers $n$, strings~$s$, and authority values $\AuthorityVal{\level}{\authpurposemeta}$.
In our semantics, we denote the base type (integer, string, or authority) of a base value $\basemeta$ as $\typeof{\basemeta}$,
and we furthermore assign a predetermined type for each program variable such that $\typeof{x}$ denotes the type of variable $x$.
The types of variables are static and cannot be changed during the execution.
Fig.~\ref{fig:semantics:expressions} presents the rules for expression evaluation and Fig.~\ref{fig:semantics} presents the command evaluation rules for our language.
Note how a \code{tini} statement reduces to the sequential composition of its argument and a special \code{pcdecl} command.  The syntactic structure imposed by the \code{tini} blocks ensures that the use of \code{pcdecl} is always well-bracketed since the \code{pcdecl}-command is not part of the surface language.
At runtime, the expanded \code{pcdecl}-commands exhibit a stack-like behavior reminiscent of pc-stacks in other
monitor designs from the literature.

The monitor is inherently progress-sensitive: barring any \code{pcdecl} commands, the $\pc$ never goes down during the execution. A reader familiar with the literature on information flow monitors may spot  deficiencies in the monitor's precision -- for example, it rejects program $(\IfThenElse{h}{\Skip}{\Skip}); \Assign{l}{0}$. This simple monitor is picked for the purpose of exposition to allow us to focus on the presentation of the security condition and the soundness proof in Section~\ref{sec:security-condition}. We further note that while it is possible to add extra precision to this monitor, unlike progress-insensitive monitors that benefit from hybrid analysis, it is difficult to avoid pc creep in progress-sensitive monitors.

\begin{figure}
\framebox[0.98\width]{
  \begin{mathpar}
    \inferrule
    % [E-Lit]
    {
    }{
      \configTwo{\basemeta}{\mem} \yields \LabeledValue{\basemeta}{\bot}
    }
    \and
    \inferrule
    % [E-Lookup]
    {
      \mem(x) = \basemeta
    }{
      \configTwo{x}{\mem} \yields \LabeledValue{\basemeta}{\levelof{x}}
    }
    \and
    \inferrule
    % [E-BinOp]
    {
      \configTwo{\expr_1}{\mem} \yields \LabeledValue{\basemeta_1}{\level_1} \\
      \configTwo{\expr_2}{\mem} \yields \LabeledValue{\basemeta_2}{\level_2} \\
      \typeof{\basemeta_1} = \typeof{\basemeta_2} \\
      \basemeta = \basemeta_1 \oplus \basemeta_2
    }
    {
      \configTwo{\expr_1 \oplus \expr_2 }{\mem} \yields \LabeledValue{\basemeta}{\level_1 \sqcup \level_2}
    }
    \and
    \inferrule
    % [E-Attenuate]
    {
      \configTwo{\expr_1}{\mem} \yields \LabeledValue{\AuthorityVal{\level_{\mathit{auth}_1}}{p_1}}{\level} \\
      {\level_{\mathit{auth}_2}} \flowsto {\level_{\mathit{auth}_1}} \\
      p_2 \leq p_1
    }
    {
      \configTwo{\Attenuate {\expr_1} { ({\level_{\mathit{auth}_2}}, {p_2}) }} {\mem} \yields \LabeledValue{\AuthorityVal{\level_{\mathit{auth}_2}}{p}}{\level}
    }

  \end{mathpar}
  }
  \caption{Semantics of evaluating expressions\label{fig:semantics:expressions}}
\end{figure}

\begin{figure*}
\framebox[0.99\width]{
  \begin{mathpar}
    \inferrule
    % [S-SKIP]
    {~}
    {\configThree{\Skip}{m}{\pc} \stepsto \configThree{\Stop}{m}{\pc} }
    \and
    \inferrule
    % [S-ASSIGN]
    {
      \configTwo{e}{m} \yields  \LabeledValue{v}{\level_e}
      \and
      \typeof{x} = \typeof{v}
      \and
      \pc \sqcup \level_e \flowsto \levelof{x}
    }
    {\configThree{\Assign{x}{e}}{m}{ \pc}
      \stepswith{\EventAssign{x}{v}} \configThree{\Stop}{m [x  \mapsto v] }{\pc } }
    \and
    \inferrule
    % [S-SEQ1]
    {
      \configThree{\cmd_1}{m}{\pc} \stepswith{\eventMeta} \configThree{\Stop}{m'}{\pc'}
    }
    {
      \configThree{\cmd_1; c_2}{m}{\pc} \stepswith\eventMeta \configThree{\cmd_2}{m'}{\pc'}
    }
    \and
    \inferrule
    % [S-SEQ2]
    {
      \configThree{\cmd_1}{m}{\pc'} \stepswith\eventMeta \configThree{\cmd'_1}{m'}{\pc'} \and c'_1 \neq \Stop
    }{
      \configThree{\cmd_1; c_2}{m}{\pc} \stepswith\eventMeta \configThree{\cmd_1';c_2}{m'}{\pc'}
    }
    \and
    \inferrule
    % [S-IF]
    {\configTwo{e}{m} \yields \LabeledValue{\basemeta}{\level} \and
      i = {\begin{cases}
          2 &\text{if } \basemeta = 0 \\
          1 &\text{otherwise}
        \end{cases}}}
    {
      \configThree{\IfThenElse{e}{\cmd_1}{\cmd_2}}{m}{\pc}
      \stepsto
      \configThree{\cmd_i}{m}{\pc \sqcup \level}
    }
    \and
    \inferrule
    % [S-WHILE]
    {~}
    {
      \configThree{\While{e}{\cmd}}{m}{\pc} \stepsto
      \configThree{\IfThenElse{e}{\cmd;\While{e}{\cmd}}{\Skip}}{m}{\pc}
    }
    \and
     \inferrule%[S-DECLASSIFY]
     {
       \configTwo{e_{\mathit{auth}}}{m} \yields
       \LabeledValue{\AuthorityVal{\level_{\mathit{auth}}}{1} }{\level'}
       \and
       \configTwo{e}{m} \yields  \LabeledValue{v}{\level_{\mathit{from}}}
       \and
       \typeof{x} = \typeof{v}
       \and 
       \level' \flowsto \pc 
       \and \\
       \level_{\mathit{to}} \sqcup \pc \flowsto \levelof{x}
       \and
       \level_{\mathit{from}} \flowsto \level_{\mathit{to}} \sqcup \level_{\mathit{auth}}
       \and
       \eventMeta = \EventDeclassify{x}{\level_{\mathit{auth}}}{\level_{\mathit{from}}}{\level_{\mathit{to}}}
       \and
       m' = {m [x \mapsto v ] }
     }{
       \configThree{\Declassify{x}{\expr_{\mathit{auth}}}{\level_{\mathit{to}}}{\expr}}{m}{\pc}
       \stepswith{\eventMeta}
       \configThree{\Stop}{m'}{\pc}
     }
    \and
    \inferrule
    %[S-PINI]
    {
      \configTwo{e}{m} \yields \LabeledValue{ \AuthorityVal{\level}{p} }{\level'}
      \and
      p \geq 0
      \and
      \level' \flowsto \pc
      \and
      \pc \flowsto \level_{\mathit{to}}
    }{
      \configThree{ \Tini{\eta}{\level_{\mathit{to}}}{\expr}{\cmd} } {m}{\pc}
      \stepswith {%\EventTiniEnter{\eta}
       } \configThree{ c ; \PCDecl{\eta}{\level}{{\level_{\mathit{to}}}} } { m }{\pc}
    }
    \and
    \inferrule
    % [S-PCDECL]
    {
      \pc_{\mathit{from}}  \flowsto \pc_{\mathit{to}} \sqcup  \level_{\mathit{auth}}
      \and
      \eventMeta = \EventTiniExit{\eta}{\level_{\mathit{auth}}} { \pc_{\mathit{from}} }{ \pc_{\mathit{to}} }
    }
    {\configThree{\PCDecl{\eta}{\level_{\mathit{auth}}}{\pc_{\mathit{to}}}}{m}{\pc_{\mathit{from}}}
      \stepswith{\eventMeta}
      \configThree{\Stop}{m}{\pc_{\mathit{to}}}
    }
    \and
    \inferrule
    % [S-EVAL]
    {
      \configTwo {\expr}{m} \yields \LabeledValue{s}{\level}\\
      c = \parseCmd{s}\\
      \Vars{c} \subseteq \{x_1, \ldots, x_n\} \\
      \evalFree{c}
    }
    {
      \configThree { \Eval \expr \{x_1, \ldots, x_n\} } { m } {\pc}
      \stepsto
      \configThree { c } { m } {\pc \sqcup \level}
    }
  \end{mathpar}
  }
  \caption{Monitored operational semantics \label{fig:semantics}}
\end{figure*}

\paragraph{Declassifications}
Our language has two different constructs for downgrading:
one for downgrading values (\code{decl}),
and one for downgrading the termination of a region of the program (\code{tini}).
We include two constructs to highlight differences and parallels between the two kinds of declassifications.
Both constructs reveal information by design, but in different ways.
Whereas declassification is a way for the programmer to indicate that an otherwise secret value is public,
the \code{tini} constructs allows the programmer to indicate that a program block (identified by a unique tag $\eta$) should be treated in a progress-insensitive way, which means that the information about the termination of the block is public.
In the jargon of information flow control systems, this exactly amounts to lowering the $\pc$-label at the end of the block.

\paragraph{Authority}
Our language restricts the use of declassifications via a capability-like mechanism that we refer to as \emph{authority}~\cite{Jif}.
Given  a value at
level~$\level_{\mathit{from}}$, an authority of
level $\level_{\mathit{auth}}$ permits a declassification to level
$\level_{\mathit{to}}$ if
$\flowsTo{\level_{\mathit{from}}}{  \level_{\mathit{to}} \sqcup \level_{\mathit{auth}}}$.
At run-time, an authority value $\AuthorityVal {\level} {\authpurposemeta}$ consists of an authority level $\level$ and a purpose bit $\authpurposemeta$.
The purpose bit~1 means that the authority can be used for general purpose declassification,
while the purpose bit~0 means that the authority can only be used for \code{tini}-statements.
For example, assuming that variable $\mathit{auth}_{\mathit{M}}$ contains the value $\AuthorityVal{\mathit{M}}{1}$,
the language allows the declassification
$$
\Declassify{l}{\mathit{auth}_{\mathit{M}}}{L}{m}
$$
but not
$$
\Declassify{l}{\mathit{auth}_{\mathit{M}}}{L}{h}
$$

\paragraph{Attenuate and running untrusted code}
The only way to create an authority value in the language is by attenuation of another authority value.
Initially, the special variable $\RootAuth$ contains the full authority $\AuthorityVal {\top} {1}$.
Our language contains primitives for restricting the access, level, and purpose of authority,
namely \code{attenuate} and \code{eval}.

For example, $\configTwo{\Attenuate {\RootAuth} { ({M}, 0) }} {\mem}$ evaluates to a value $\LabeledValue{\AuthorityVal{M}{0}}{\bot}$ that can only be used for declassifying progress up to level $M$.
For running untrusted code, we provide an $\code{eval}$ command that takes a string $s$ and a set of variables $\{x_1, \dots, x_n\}$.
The semantics of \code{eval} is, that it parses the string to a command $c$ (denoted $c = \parseCmd{s}$) under the condition that $c$ is only allowed to use variables explicitly mentioned in $\{x_1, \dots, x_n\}$ and must not contain nested \code{eval}s.
In this way, our $\code{eval}$-command can be seen as a ``poor man's''-scoping, which we capture in the following Lemma:
\begin{lemma}[$\code{eval}$ memory safety]
  \label{lemma:eval-scoping}
  Suppose
  $\configThree{\Eval e X}{m}{\pc} \stepsmanywith{\events} \configThree {c'} {m'} {\pc'}$.
  Then it holds for all $s$ where $x \in X \implies  m(x) = s(x)$ that
  $$\configThree{\Eval e X}{s}{\pc} \stepsmanywith{\events} \configThree {c'} {s'} {\pc'}$$
  and
  $$x \in X \implies  m'(x) = s'(x)$$
\end{lemma}
\begin{proof}
  By induction in the program resulting from $\parseCmd{s}$ using that no variables except those occurring in $X$ is used.
\end{proof}

The combination of \code{eval} and \code{attenuate} allows us to attenuate the root-authority by storing
 it in some variable, e.g., $x$, and run untrusted code while only permitting access to $x$.
For example, we may restrict declassifications 
in the evaluation of the command stored in variable $\mathit{m_{code}}$
up to level~$M$ as follows.
\begin{align*}
  &\Assign{\mathit{auth}_{\mathit{M}}}{\Attenuate {\RootAuth} {(M, 1)}}; \\
  &\Eval {\mathit{m_{code}}} {\{\mathit{auth}_{\mathit{M}}, l_1, l_2, m_1, m_2, h_1, h_2\}}
\end{align*}
Note that the program in $\mathit{m_{code}}$ may access high variables $h_1$ and $h_2$ but cannot declassify them since it does not have access to sufficient authority.

%\todo{show examples of running untrusted code}

\paragraph{\code{tini}-blocks}
The \code{tini}-construct allows us to embed progress-insensitive code in an otherwise progress-sensitive setting.
To give some intuition about the \code{tini}-construct, suppose we have the following program
that loops if a variable of level $H$ is positive; or makes an assignment at level $L$ otherwise:
\begin{align*}
  &\While{h > 0 }{\Skip}\\
  &\Assign{l}{0}
\end{align*}

This program is acceptable in a progress-insensitive setting, but is rejected by progress-sensitive security conditions,
since the assignments to $\mathit{l}$ leaks information about the reachability of the join-point.
The \code{tini} construct allows us to embed such code in a progress-sensitive setting by explicitly declassifying the reachability of the end of the block.
Just like regular declassification, the \code{tini}-block also requires an authority argument.
Hence, the example above can be written instead as:
\begin{align*}
  &\Tini{\eta}{L}{\RootAuth}
    {\\&\qquad \While{h > 0 }{\Skip}};\\
  &\Assign{l}{0}
\end{align*}

The design of the \code{tini} block is inspired by similar constructs in large-scale information flow systems:
Jif~\cite{Jif} implements \code{pc}-declassification by a single command for declassifying the $pc$-label although the syntax does not limit the scope of the progress that is declassified.
HiStar~\cite{Histar} implements a similar thing through ``untainting'' gates that can be restricted to only untaint the control flow.

Attenuation of the purpose can be used in conjunction with \code{eval} and the \code{tini} block.
Revisiting the news widget example from Section~\ref{sec:introduction}, the trusted code may evaluate the widget by passing it access to an attenuated authority.
To bring the example closer to the language we have presented, we let $\code{receive}$ fetch the untrusted widget code from a network connection and run it by using $\code{eval}$:
\begin{align*}
  & \Assign{\mathit{untrustedWidget}} {\code{receive\ } \mathit{newsfeed\_server\_url}} ; \\
  & \Assign {\mathit{userFavTopic}} {\code{"Politics"}}; \\
  & \Assign {\mathit{authNews}} {\Attenuate {\RootAuth} {(\mathit{newslev}, 0)}}; \\
%  & \Eval {\mathit{untrustedWidget}} {\{\mathit{userFavTopic}, \mathit{authNews} \}} \\
    &\Tini{\eta}{\bot}{\mathit{authNews}}
    {\\&\qquad   \Eval {\mathit{untrustedWidget}} {\{\mathit{userFavTopic}\}} }
\end{align*}

\section{Security condition}
\label{sec:security-condition}

This section presents a security definition for embedding \code{tini}-blocks when the baseline security is progress-sensitive.

\subsection{Auxiliary definitions}
We use the knowledge-based~\cite{ChangeInKnowledge} approach to define
our security condition.
The high-level idea behind the approach is that we consider an attacker that can observe the execution of the program
and define the knowledge that such attacker obtains as the set of memories that are consistent with seeing the execution up to this point.
The security condition is defined as a bound on how much the knowledge is allowed to change at each step of the execution.

To define such bounds, we first define what it means for memories to be equivalent and define which execution steps are visible to the adversary.

In the following, 
we write $m \loweq{\level} s$ to denote that two memories are equal up to $\level$ (Definition~\ref{def:mem-equiv} below), and $\tproj{\events}{\level}$ to denote a filtering of the trace $\events$ that only includes the events that are observable at level $\level$ (Definition~\ref{def:trace-proj} below).

\begin{definition}[Memory equivalence]\label{def:mem-equiv}
  Two memories $m$ and $s$ are equivalent up to level $\ell$, written $m \loweq{\level} s$,
  if $\dom{m} = \dom{s}$ and it holds that for all $x \in \dom{m}$,
  $$
  \levelof{x} \sqsubseteq \level \implies m(x) = s(x)
  $$
\end{definition}

We define level of an event, denoted $\levelof{\eventMeta}$, as the level of the updated variable for assignment and declassify events,
level $\ell_{\mathit{to}}$ for tini events $\EventTiniExit{\eta}{\level}{\level_{\mathit{from}}}{\level_{\mathit{to}}}$, and $\top$ otherwise:
\begin{align*}
  \levelof{\EventEmpty} &= \top \\
  \levelof{\EventAssign{x}{\_}} &= \levelof{x} \\
  \levelof{\EventDeclassify{x}{\_}{\level}{\_} \mid} &= \levelof{x} \\
  \levelof{\EventTiniExit{\eta}{\_}{\_}{\level_{\mathit{to}}}} &= {\level_{\mathit{to}}}
\end{align*}

\begin{definition}[Trace filtering]\label{def:trace-proj}
  The filtering of a trace $\events$ at level $\level$ written $\tproj{\events}{\level}$ is defined as
  \begin{align*}
    \tproj{[]}{\level} &= [] \\
    \tproj{\events' \cdot \alpha}{\level} &=
                                            \begin{cases}
                                              \tproj{\events'}{\level} \cdot \alpha &\text{if } \levelof{\alpha} \sqsubseteq \level \\
                                              \tproj{\events'}{\level} &\text{otherwise}
                                            \end{cases}
  \end{align*}
\end{definition}

We use the above to define two technical definitions of knowledge.
First, we define attacker knowledge which defines the knowledge of an attacker observing a trace $t$.
\begin{definition}[Attacker knowledge~\cite{Gradualrelease}]
  Given a program $c$, initial memory $m$, initial program counter level $\pc$,
  such that
  $\configThree{\cmd}{m}{\pc} \stepsmanywith{\events} \configThree{\cmd'}{m'}{\pc'}$,
  define \emph{attacker knowledge at level~$\level_\ladv$} to be the set of memories $m'$ that are
  consistent with the observations of the adversary:
  \begin{multline*}
    \knowledge(c, \mem, \events, \level_\ladv ) \defn
    \{ \mem' \mid \mem \loweq{\level_\ladv} \mem' \land \\ \configThree{\cmd}{m'}{\pc} \stepsmanywith{\events'} \configThree{\cmd''}{m''}{\pc''}
    \land \tproj{\events'}{\level_\ladv} = \tproj{\events}{\level_\ladv}
    \}
  \end{multline*}
\end{definition}

We can now use this definition as a building block for defining security conditions.
We can, for example, define progress-sensitive noninterference as follows:
\begin{definition}[Progress-sensitive noninterference]
  Given a program~$\cmd$, initial memory $\mem$ and initial program counter label $\pc$ such that
  $$
  \configThree{\cmd}{\mem}{\pc}
  \stepsmanywith{\events \cdot \alpha}
  \configThree{\cmd'}{\mem'}{\pc'}
  $$
  the run satisfies \emph{progress-sensitive noninterference} if it holds that
  for all $\level_{\adv}$, if $\levelof{\eventMeta} \flowsto \level_{\adv}$ then
  $$
  \knowledge(c, m, t \cdot \alpha, \level_{\adv}) \supseteq \knowledge(c, m, t, \level_{\adv})
  $$
\end{definition}
Note how this definition bounds the knowledge from seing~$t \cdot \alpha$ with the knowledge of seeing $t$.
This essentially means that all the memories that the attacker considered possible when seeing~$t$ are still considered possible after also observing the event $\alpha$.
Note that his is a very strong security condition. To define more lenient conditions, we use another building block: the progress knowledge.

\begin{definition}[Progress knowledge~\cite{AskarovSabelfeld}]
  Given a program $c$, initial memory $m$, initial program counter level $\pc$,
  such that
  $\configThree{\cmd}{m}{\pc} \stepsmanywith{\events} \configThree{\cmd'}{m'}{\pc'}$,
  define \emph{progress knowledge at level~$\level_\ladv$} to be the set of memories $m'$ that are
  consistent with the knowledge up to $\events$ followed further by one more event:
  \begin{multline*}
    \progressknowledge(c, \mem, \events, \level_\ladv ) \defn
    \{ \mem' \mid \mem \loweq{\level_\ladv} \mem' \land \\ \configThree{\cmd}{m'}{\pc} \stepsmanywith{\events'} \configThree{\cmd''}{m''}{\pc''}
    \land \tproj{\events'}{\level_\ladv} = \tproj{\events}{\level_\ladv } \cdot \eventMeta
    \}
  \end{multline*}
\end{definition}

The above allows us to express the standard progress-insensitive noninterference:
\begin{definition}[Progress-insensitive noninterference]
  Given a program~$\cmd$, initial memory $\mem$ and initial program counter label $\pc$ such that
  $$
  \configThree{\cmd}{\mem}{\pc}
  \stepsmanywith{\events \cdot \alpha}
  \configThree{\cmd'}{\mem'}{\pc'}
  $$
  the run satisfies \emph{progress-insensitive noninterference} if it holds that
  for all $\level_{\adv}$, if $\levelof{\eventMeta} \flowsto \level_{\adv}$ then
  $$
  \knowledge(c, m, t \cdot \alpha, \level_{\adv}) \supseteq \progressknowledge(c, m, t, \level_{\adv})
  $$
\end{definition}
Here, the knowledge of an attacker that observes $t \cdot \alpha$ is bounded by the progress knowledge from seeing just $t$.
This exactly captures that the attacker is allowed to rule out the the memories that do not make progress. 

\subsection{Progress-sensitive security with declassification and locally-bound progress-insensitivity}
Armed with the above definitions, we define our main security condition as follows.

\begin{definition}[Progress-sensitive security with declassification and locally-bound progress-insensitivity] \label{def:bpni}
  Given a program~$\cmd$, initial memory $\mem$ and initial program counter label $\pc$ such that
  $$
    \configThree{\cmd}{\mem}{\pc}
    \stepsmanywith{\events \cdot \alpha}
    \configThree{\cmd'}{\mem'}{\pc'}
  $$
  define the run as \emph{secure} if it holds that
  for all $\level_{\adv}$, if $\levelof{\eventMeta} \flowsto \level_{\adv}$ then
  \begin{enumerate}%[label=(\Roman*)]
  \item if $\alpha = \EventDeclassify{\_}{\ell_{\mathit{auth}}}{\ell_{\mathit{from}} }{\ell_{\mathit{to}}}$
    \label{item:1}
    then it should hold that:
    \begin{enumerate}
    \item $\progressknowledge(\cmd, \mem, \events, \level_\adv) \supseteq \knowledge(\cmd, \mem, \events, \level_\adv)$, and
      \label{item:1.1}
    \item $\knowledge(\cmd, \mem, \events \cdot \alpha, \level_\adv) \supseteq \knowledge(\cmd, \mem, \events, \level_\mathit{auth} \sqcup \level_{\adv})$
      \label{item:1.2}
    \end{enumerate}
  \item if $\alpha = \EventTiniExit{\eta}{\ell_{\mathit{auth}}}{\ell_{\mathit{from}}}{\ell_{\mathit{to}}}$
    then it should hold that:
    \label{item:2}
    \begin{enumerate}
    \item $\knowledge(\cmd, \mem, \events \cdot \alpha, \level_\adv) \supseteq \progressknowledge(\cmd, \mem, \events, \ell_\adv)$
    \item $ \progressknowledge(\cmd, \mem, \events, \ell_\adv) \supseteq \knowledge(\cmd, \mem, \events, \level_{\mathit{auth}} \sqcup \level_{\adv})$
    \end{enumerate}
  \item otherwise, it should hold that:
    \label{item:3}
    \begin{align*}
      \knowledge(\cmd, \mem, \events \cdot \alpha, \level_\adv)
      \supseteq \knowledge(\cmd, \mem, \events, \ell_\adv) \\
    \end{align*}
  \end{enumerate}
\end{definition}

The security condition specifies what information the attacker may learn from observing the program events. The baseline of progress-sensitive security is captured in item~\ref{item:3} of the definition stating that the attacker learns nothing from non-declassify events. This rules out many standard examples of direct and indirect flows, as well as the termination leaks such as $$l= 0; (\code{while}\ h > 0\  \code{do}\ \code{skip}); l= 1$$

The other two items weaken the baseline as follows. For declassifications (item~\ref{item:1}) we have two clauses:
Clause~\ref{item:1.1} says that reachability of the declassification conveys no knowledge to the attacker.
Observe that this is expressed as a bound on the progress knowledge!
This clause rules out programs such~as 
$$l= 0; (\code{while}\ h > 0\  \code{do}\ \code{skip}); \Declassify{l}{\mathit{auth}_{\mathit{H}}}{L}{h}$$
that leak via termination without a \code{tini}-statement.

Clause~\ref{item:1.2} specifies an upper bound on the information the attacker learns from the event to be no more the knowledge at level~$\level_{\mathit{auth}} \sqcup \level_\adv$ before the event.
This clause has a flavor of language-based intransitive noninterference~\cite{intransitive}, because it does not otherwise
bound what information from the permitted level is declassified. For example, assuming $\mathit{auth}_M$ and $\mathit{auth}_H$ are
authorities with purpose bit one, this definition accepts the program
\begin{align*}
& \Declassify{m}{\mathit{auth}_H}{M}{h}; \\ 
& \Declassify{l}{\mathit{auth}_M}{L}{m} 
\end{align*}
Both declassifications above are allowed. At the time of the second declassification,
the adversary at $L$ learns the original value of $h$ despite only using the authority of $M$. This is accepted because 
the earlier declassification of $h$ to $m$ happened with sufficient authority.

Clause~\ref{item:1.2} does not regulate exactly \emph{what} information from the level of $\level_{\mathit{auth}} \sqcup \level_\adv$ may be declassified; however, prior work on using knowledge-based conditions for further
constraining what and where to declassify can be easily applied here in an orthogonal manner~\cite{AskarovSabelfeld,chudnov2018assuming}.

For \code{tini}-events (item~\ref{item:2}), we also have two constraints.
The first constraint corresponds to standard progress-insensitive noninterference~\cite{tini}: knowledge of the event must reveal no more than knowledge of the event's existence.
The second constraint is interesting, because it specifies an upper bound on the information leaked by the termination to be no more than the knowledge at level $\level_{\mathit{auth}} \sqcup \level_\adv$ before the event. This is again expressed as a bound on progress knowledge.
This clause rules out programs with insufficient authority for the $\pc$-declassi\-fication such as
$$
l=0; (\Tini{\eta}{L}{\mathit{auth}_{\mathit{M}}}{\code{while}\ h > 0\  \code{do}\ \code{skip}}); l = 1
$$

The definition accepts programs that use \code{tini} blocks as long as the authority for the $\pc$-declassification is sufficient. This
includes nested \code{tini} blocks.  The following program is accepted.
\begin{align*}
& l = 0;\\
& \Tini{\eta_1}{L}{\mathit{auth}_{\mathit{M}}}{ \{ } \\
& \quad \code{if}\ m > 0\ \code{then} \\
& \quad \quad \Tini{\eta_2}{M}{\mathit{auth}_{\mathit{H}}}{} \\
& \quad \quad \quad \code{while}\ h > 0\  \code{do}\ \code{skip} \\
& \quad \code{else\ skip}\ \}  ; \\
& l = 1
\end{align*}

\subsection{A note on the design of item~\ref{item:2}}
\label{twoclauses}
For the simple language of this section, the two clauses of item \ref{item:2} can be simplified to require that for $\alpha = \EventTiniExit{\eta}{\ell_{\mathit{auth}}}{\ell_{\mathit{from}}}{\ell_{\mathit{to}}}$ it must hold that
$$
\knowledge(\cmd, \mem, \events \cdot \alpha, \level_\adv) \supseteq \knowledge(\cmd, \mem, \events, \level_{\mathit{auth}} \sqcup \level_{\adv})
$$
We opted to present the definition without this simplification, because 
in more realistic settings, this simplification is dangerous and leads to occlusion. 

The simplification is possible in our language, because 
 \code{pcdecl} events are attacker-observable and convey little information other than their 
 reachability, thanks to syntactically enforced well-bracketedness of \code{tini}/\code{pcdecl} commands.\footnote{These conveniences help us minimize technical clutter in the paper.}

However, in reality, it may be unfair to assume that attacker observes internal events such as
 \code{pcdecl}. Suppose indeed that \code{pcdecl} has no manifestation 
in the attacker-observable projection of the trace. How would we need to change Definition~\ref{def:bpni} to accommodate this? 
One option is to rephrase item~\ref{item:2} of Definition~\ref{def:bpni} so that event~\code{\eventMeta} refers to \emph{the first observable event after executing \code{pcdecl}}. But such events can communicate more than a unit of information,
as in the program below.
\begin{align*}
    & \Tini{}{L}{ \RootAuth }{ \{ \Skip \} } \\
    & \IfThenElse{h > 0}{ \Assign{l}{0}}{\Assign{l}{1}}
\end{align*}
This program would reduce to 
$$\PCDecl{} {\RootAuth} {L}; \IfThenElse{h > 0}{ \Assign{l}{0}}{\Assign{l}{1}}$$ 
Here, the first event after \code{pcdecl} is one of the low assignments. The approach of the simplified definition accepts this program because
it mistakenly applies the declassification condition to reveal the choice of the high branch.
On the other hand, the two-clause approach that
explicitly constraints the progress knowledge rejects this program.

\subsection{Soundness of the enforcement}

Next, we formally connect the monitoring semantics of Section~\ref{sec:language} with 
Definition~\ref{def:bpni}.
We do this by showing the following statement:
\begin{theorem}[Soundness of the monitoring semantics]\label{thm:soundness-bpni}
  Given a program $c$, memory $\mem$, and level $\pc$ then
  all runs
  $\configThree{\cmd}{m}{\pc} \stepsmanywith{\events}
  \configThree{\cmd'}{m'}{\pc'}$ satisfy
  Definition~\ref{def:bpni}.
\end{theorem}

\begin{figure}
\framebox[0.98\width]{
  \begin{mathpar}
    \inferrule
    % [Bridge-Public]
    {
      \cfgmeta \stepswith{\alpha} \cfgmeta'\\
      \obs{\levelof\alpha}{\ell}
    }
    {
      \cfgmeta \lbridgestep{0}{\ell}{\alpha} \cfgmeta'
    }
    \and
    \inferrule
    % [Bridge-Stop]
    {
      \cfgmeta \stepswith{\beta} \configThree{\Stop}{m}{\pc} \\
      \nobs{\levelof\beta}{\ell}
    }
    {
      \cfgmeta \lbridgestep{0}{\ell}{\beta} \configThree{\Stop}{m}{\pc}
    }
    \and
    \inferrule
    % [Bridge-Multi]
    {
      \cfgmeta_1 \stepswith{\beta} \cfgmeta_2\\
      \nobs{\levelof{\beta}}{\ell}\\
      \cfgmeta_2 \lbridgestep{n}{\ell}{\alpha} \cfgmeta_3\\
    }
    {
      \cfgmeta_1 \lbridgestep{n+1}{\ell}{\alpha} \cfgmeta_3
    }
  \end{mathpar}
}
  \caption{Bridge-step relation \label{fig:bridge-step}}
\end{figure}

To get some intuition about the proof, let us think how classical noninterference proofs usually proceed. The security invariant of such proofs boils down to the reasoning along the lines of ``a pair of low-equivalent configurations that each emit attacker-observable events transition to low-equivalent configurations plus the attacker cannot discriminate between the two events.'' Note how low-equivalence is used in both the precondition and the post-condition of such a statement. For declassification, we need to weaken the invariant, which is typically done by strengthening the precondition to relate fewer configurations. Set-theoretically, this strengthening corresponds to picking a relation that is smaller than low-equivalence. Exactly how small is an important design criterion that is dictated by the top-level security requirement such as our Definition~\ref{def:bpni}. 
One challenge that we have encountered in the proof is finding the right equivalence relation for the precondition that is compositional in the applications of the inductive hypothesis.
Our solution to this challenge is to engineer relations that are smaller than low-equivalence, subject to additional constraints we explain below.

First, we define an auxiliary relation that characterizes the intuition of ``configuration emitting an attacker-observable event.'' We call this relation \emph{bridge-step}. Operationally it is defined as a relation between two configurations where the first configuration reaches the second one by taking $n$ intermediate ``secret'' steps (without producing any observable events) and then either emits an observable step or terminates. 
This relation is shown in Fig.~\ref{fig:bridge-step}.
The security intuition behind the bridge relation is that the attacker only observes the configurations related by the bridge relation. Hence, we formulate our security invariant around that relation.

We furthermore define \emph{indistinguishability restriction} $\indnarr{\indmeta}{c}{\pc}{\ell}{\alpha_1 \dots \alpha_k}$ as the restriction of the relation $\indmeta$ to only contain all pairs of the memories that can emit $\alpha_1 \dots \alpha_k$ (in that order) when evaluating $c$ using initial program-counter $\pc$.
To formally define $\indnarr{\indmeta}{c}{\pc}{\ell}{\alpha_1 \dots \alpha_k}$, we introduce another auxiliary definition that synchronizes two bridge-step runs on a list of events.
The synchronized bridge has the effect of demanding that two runs proceed in lock-step w.r.t to their individual bridge-steps.
The rules for synchronized bridging can be seen in Fig.~\ref{fig:synchronized-bridging}.

\begin{figure}
\framebox[0.98\width]{
  \begin{mathpar}
    \inferrule{
      \configThree{c}{m}{\pc} \lbridgestep{n}{\ell}{\alpha} \configThree{c'}{m'}{\pc'}
      \and
      \configThree{c}{s}{\pc} \lbridgestep{n'}{\ell}{\alpha} \configThree{c'}{s'}{\pc'}
    }{
      \syncConfig{c}{m}{s}{\pc}
      \syncbridge{\ell}{\alpha}
      \syncConfig{c'}{m'}{s'}{\pc'}
    }
    \\
    \inferrule{
      k > 1 \and
      \syncConfig{c}{m}{s}{\pc} \syncbridge{\ell}{\alpha_1}
      \syncConfig{c'}{m'}{s'}{\pc'}  \\
      \syncConfig{c'}{m'}{s'}{\pc'}
      \syncbridge{\ell}{\alpha_2 \dots \alpha_k}
      \syncConfig{c''}{m''}{s''}{\pc''}
    }{
      \syncConfig{c}{m}{s}{\pc} \syncbridge{\ell}{\alpha_1 \dots \alpha_k}
      \syncConfig{c''}{m''}{s''}{\pc''}
    }
  \end{mathpar}
  }
  \caption{Synchronized bridging}
  \label{fig:synchronized-bridging}
\end{figure}

We can now define indistinguishability restriction as per Definition~\ref{def:indnar} below.

\begin{definition}[Indistinguishability restriction $\indnarr{\indmeta}{c}{\pc}{\ell}{\eventMeta}$]
  \label{def:indnar}
  Consider a potentially empty sequence of events $\alpha_1 \dots \alpha_k$. Define the relation
  $m \mathrel{\indnarr{\indmeta}{c}{\pc}{\ell}{\alpha_1 \dots \alpha_k}} s$ as follows:
  \begin{mathpar}
    \inferrule{
      m \mathrel{I} s
    }{
      m \indnarr{\indmeta}{c}{\pc}{\ell}{\emptylist} s
    }
    \and
    \inferrule{
      m \mathrel{I} s
      \and
      \syncConfig{c}{m}{s}{\pc} \syncbridge{\ell}{\alpha_1 \dots \alpha_k}
      \syncConfig{c'}{m'}{s'}{\pc''}
    }{
      m \mathrel{\indnarr{\indmeta}{c}{\pc}{\ell}{\alpha_1 \dots \alpha_k}} s
    }
  \end{mathpar}
\end{definition}

With all this auxiliary infrastructure, we  now state the operational definition of security.

\begin{lemma}[Security for monitored evaluations]
  \label{lem:NI-bridge}
  Suppose
  $\configThree{c}{m}{\pc} \lbridgestep{n}{\ell_\adv}{\alpha} \configThree{c'}{m'}{\pc'}$.
  Then the following holds:
  \begin{enumerate}% [label=(\Roman*)]
  \item \textbf{if $\alpha =
      \EventDeclassify{x}{\level_{\mathit{auth}}}{\level_{\mathit{from}}}{\level_{\mathit{to}}}$ and $\levelof{x} \flowsto \level_\adv$:}\\ \label{bridge-item:1}
    Let
    $$I = \indnarr{\loweq{\level_{\mathit{auth}} \sqcup \level_{\adv}}}{c}{\pc}{\level_{\mathit{auth}} \sqcup \level_{\adv}}{\beta_1 ,\ldots, \beta_j}$$
    where
    \begin{align*}
      \configThree{c}{m}{\pc} & \lbridgestep{i_1}{\level_{\mathit{auth}} \sqcup \level_{\adv}}{\beta_1} \configThree{c_1}{m_1}{\pc_1}
       \lbridgestep{i_2}{\level_{\mathit{auth}} \sqcup \level_{\adv}}{\beta_2}\\
       &\ldots
       \lbridgestep{i_j}{\level_{\mathit{auth}} \sqcup \level_{\adv}}{\beta_j} \configThree{c_j}{m_j}{\pc_j}
    \end{align*}
    such that
    $$
    \configThree{c_j}{m_j}{\pc_j} \lbridgestep{i'}{\level_{\mathit{auth}} \sqcup \level_{\adv}}{\alpha} \configThree{c'}{m'}{\pc'}
    $$
    then it holds that for all $s$ such that $m \mathrel{\indmeta} s$,
    $$
    \configThree{c}{s}{\pc} \lbridgestep{n'}{\level_\adv}{\alpha} \configThree{c'}{s'}{\pc'}
    $$
    and $m' \loweq{\level_\adv} s'$.

  \item \textbf{if $\alpha = \EventTiniExit{\eta}{\level_{\mathit{auth}}}{\_}{\level_{\mathit{to}}}$ and $\level_{\mathit{to}} \sqsubseteq \level_\adv$:}\\
    \label{bridge-item:2}
    Let
    $$I = \indnarr{\loweq{\level_{\mathit{auth}} \sqcup \level_{\adv}}}{c}{\pc}{\level_{\mathit{auth}} \sqcup \level_{\adv}}{\beta_1 ,\ldots, \beta_j}$$
    where
    \begin{align*}
    \configThree{c}{m}{\pc} & \lbridgestep{i_1}{\level_{\mathit{auth}} \sqcup \level_{\adv}}{\beta_1} \configThree{c_1}{m_1}{\pc_1}
      \lbridgestep{i_2}{\level_{\mathit{auth}} \sqcup \level_{\adv}}{\beta_2}\\
      &\ldots
      \lbridgestep{i_j}{\level_{\mathit{auth}} \sqcup \level_{\adv}}{\beta_j}
    \configThree{c_j}{m_j}{\pc_j}
    \end{align*}
    such that
    $$
    \configThree{c_j}{m_j}{\pc_j} \lbridgestep{i'} {\level_{\mathit{auth}} \sqcup \level_{\adv}} {\alpha} \configThree{c'}{m'}{\pc'}
    $$
    then it must hold that for all $s$ where $m \mathrel{\indmeta} s$ there exists $\alpha'$ such that
    $$
    \configThree{c}{s}{\pc} \lbridgestep{n'}{\level_{\adv}}{\alpha'} \configThree{c'}{s'}{\pc'}
    $$
    and $m' \loweq{\level_\adv} s'$.
  \item \textbf{if $\alpha \not = \EventTiniExit{\_}{\_}{\_}{\_}$ and $\pc' \sqsubseteq \level_\adv$:}\\
    \label{bridge-item:3}
    For all $s$ where $m \loweq{\level_\adv} s$ it holds that
    $$
    \configThree{c}{s}{\pc} \lbridgestep{n'}{\level_\adv}{\alpha} \configThree{c'}{s'}{\pc'}
    $$
    and if $\alpha$ is not a declassify event $\EventDeclassify{x}{\_}{\_}{\_}$ where $\levelof{x} \flowsto \ell_\adv$ then $m' \loweq{\level_\adv} s'$.

  \item \textbf{if $\alpha = \EventTiniExit{\eta}{\_}{\_}{\_}$ or $\pc' \not \sqsubseteq \level_\adv$:}\\
    \label{bridge-item:4}
    It holds that for all $s$ where $m \loweq{\level_\adv} s $,
    \begin{align*}
      \configThree{c}{s}{\pc}& \lbridgestep{n'}{\ell_\adv}{\alpha'} \configThree{c''}{s'}{\pc''}
      \implies\\
      & m' \loweq{\level_\adv} s' \land c' = c'' \\
      \land\, & \pc' \not \sqsubseteq \level_\adv \implies \left( \pc'' \not \sqsubseteq \level_\adv \land c' = \Stop\right) \\
      \land\, & \pc' \sqsubseteq \level_\adv \implies \left( \pc'' \sqsubseteq \level_\adv \land \alpha = \alpha'\right)
    \end{align*}
  \end{enumerate}
\end{lemma}

The indistinguishability restriction of $\loweq{\level}$, which we alluded to earlier, appears in two out of four sub-cases of the invariant. 
This is crucial in the proof when showing the clauses related to declassify-events, since this allows us to account for earlier unobservable declassifies that might become observable through the latest event.
For example, suppose we have an attacker at level $L$ and that we earlier on declassified a value $v$ from $H$ to $M$. 
Now if we later declassify that same value from $M$ to $L$, it is not enough to only assume that the initial memories satisfy memory-equivalence up to $M$ to prove the clause for declassify.
Instead, we ``replay'' the trace at a higher attacker-level  -- in this case $M$ -- which reveals the events that are otherwise only observable at this higher level -- such as declassifications from $H$ to $M$. 
We use these events to synchronize the memories, and then conclude that the two runs must declassify the same value.
This is exactly what the indistinguishability restriction condition
$\indnarr{\loweq{\level_{\mathit{auth}} \sqcup \level_{\mathit{to}}}}{c}{\pc}{\level_{\mathit{auth}} \sqcup \level_{\mathit{to}}}{\alpha_1 ,\ldots, \alpha_k}$
provides. The events $\alpha_1, \ldots, \alpha_k$ here range over the $M$-level events including $H$ to $M$ declassifications; none of these events are typically  observable by $L$.

\opt{proceedings}{
The detailed  proof of Lemma~\ref{lem:NI-bridge} can be found in the extended version of this paper~\cite{bay2020reconciling}, where we also prove 
 Theorem~\ref{thm:soundness-bpni} by showing that runs satisfying Lemma~\ref{lem:NI-bridge} satisfy Definition~\ref{def:bpni}.
}
\opt{extended}{
  The detailed  proof of Lemma~\ref{lem:NI-bridge} can be found in the Appendix, where we also prove Theorem~\ref{thm:soundness-bpni} by showing that runs satisfying Lemma~\ref{lem:NI-bridge} satisfy Definition~\ref{def:bpni}.
}

\section{Timing Sensitivity}
\label{sec:timingsensitivity}
The security condition that we present for progress-sensitive noninterference in Definition~\ref{def:bpni} can be naturally strengthened to also cover timing-sensitive noninterference.
The cautious reader might have noticed already that the monitor we present in our language is actually already enforcing this stronger notion of noninterference.
As an example, the program
\begin{align*}
  &\code{if}\ {h > 0} \\
  & \quad {\code{then}\ \Skip} \\
  & \quad {\code{else}\ \Skip; \Skip; \Skip};\\
  &\Assign {l}{0}
\end{align*}
is accepted by our progress-sensitive security condition, but is not allowed by our monitor.
This shows that our monitoring leaves room for strengthening the security condition so that examples like above are also rejected
by the definition.

To formalize this observation, we add a clock $\ts$ to our configurations $\configClock{c}{m}{\pc}{\ts}$ and timestamps to events $\tsEvent{\ts} {\eventMeta}$ such that our evaluation steps are now defined by the following rule
$$
  \inferrule
  {
    \configThree{c}{m}{\pc} \stepswith{\eventMeta} \configThree{c'}{m'}{\pc'}
  }
  {
    \configClock{c}{m}{\pc}{\ts} \stepswith{\tsEvent{\ts + 1}{\eventMeta}} \configClock{c'}{m'}{\pc'}{\ts + 1}
  }
$$

We extend the definition of when events are observable in the obvious way: a timestamped event
$\tsEvent{\ts}{\eventMeta}$ is observable at level $\ell$ if $\eventMeta$ is observable at level $\ell$.
The definitions of attacker knowledge and progress knowledge from the previous section are also ported to the new setting in a straightforward manner, noting that the initial clock value is 0.
However, we need a new knowledge combinator, that we dub \emph{clock knowledge}.
\begin{definition}[Clock knowledge]
  Given a program $c$, initial memory $m$, initial program counter level $\pc$,
  and initial timestamp $\ts$
  such that
  $\configClock{\cmd}{m}{\pc}{0} \stepsmanywith{\events} \configClock{\cmd'}{m'}{\pc'}{\ts'}$,
  define \emph{clock knowledge at level~$\level_\ladv$} to be the set of memories $m'$ that are
  consistent with the knowledge up to $\events$ followed further by one more event with timestamp $\ts$:
  \begin{multline*}
    \clockknowledge(c, \mem, \events, \level_\ladv, \ts) \defn \{ \mem' \mid \mem \loweq{\level_\ladv} \mem' \\
    \qquad\qquad \land  \configClock{\cmd}{m'}{\pc}{0} \stepsmanywith{\events'} \configClock{\cmd''}{m''}{\pc''}{\ts} \\
    \qquad\qquad\qquad  \land \tproj{\events'}{\level_\ladv} = \tproj{\events}{\level_\ladv } \cdot \tsEvent{\ts}{\eventMeta}
    \}
  \end{multline*}
\end{definition}

Observe that the clock knowledge, the progress knowledge, and the attacker knowledge are 
related by their definitions as follows:
$ \knowledge(\cmd, \mem, \events, \level_\adv) 
        \mathop{\supseteq} 
        \progressknowledge(\cmd, \mem, \events, \level_\adv) 
        \mathop{\supseteq}
        \clockknowledge(\cmd, \mem, \events, \level_\adv, \ts')         
        $.

We can now give a more precise top-level security condition:
\begin{definition}[Timing-sensitive security with declassification and locally-bound progress-insensitivity] \label{def:btni}
  Given a program~$\cmd$, initial memory $\mem$, initial program counter label $\pc$, and initial clock $\ts$ such that
  $$
    \configClock{\cmd}{\mem}{\pc}{\ts}
    \stepsmanywith{\events \cdot \tsEvent{\ts'}{\alpha}}
    \configClock{\cmd'}{\mem'}{\pc'}{\ts'}
  $$
  define the run as \emph{secure} if it holds that
  for all $\level_{\adv}$, if $\levelof{\eventMeta} \flowsto \level_{\adv}$ then
  \begin{enumerate}
  \item if $\alpha = \EventDeclassify{\_}{\ell_{\mathit{auth}}}{\ell_{\mathit{from}} }{\ell_{\mathit{to}}}$
    then it should hold that:
    \begin{enumerate}
    \item
      $
      \begin{aligned}[t]
      \clockknowledge(\cmd, \mem, \events, \level_\adv, \ts')
      &\mathop{\supseteq} \progressknowledge(\cmd, \mem, \events, \level_\adv) \\
      &\mathop{\supseteq} \knowledge(\cmd, \mem, \events, \level_\adv)
      \end{aligned}
      $

    \item $\knowledge(\cmd, \mem, \events \cdot \tsEvent{\ts'}{\alpha}, \level_\adv) \supseteq \knowledge(\cmd, \mem, \events,\ell_{\mathit{auth}} \sqcup \ell_{\adv})$
    \end{enumerate}
  \item if $\alpha = \EventTiniExit{\eta}{\_}{\ell_{\mathit{from}}  }{\_}$
    then it should hold that:
    \begin{enumerate}
    \item $\knowledge(\cmd, \mem, \events \cdot \tsEvent{\ts'}{\eventMeta}, \level_\adv) \supseteq \clockknowledge(\cmd, \mem, \events, \ell_\adv, \ts')$
    \item $\clockknowledge(\cmd, \mem, \events, \ell_\adv, \ts') \supseteq \knowledge(\cmd, \mem, \events,\ell_{\mathit{auth}} \sqcup \ell_{\adv})$
    \end{enumerate}
  \item otherwise it should hold that:
    \begin{align*}
      \knowledge(\cmd, \mem, \events \cdot \tsEvent{\ts'}{\eventMeta}, \level_\adv)
      \supseteq \knowledge(\cmd, \mem, \events, \ell_\adv) \\
    \end{align*}
  \end{enumerate}
\end{definition}

Observe that Clause 3 of the above definition now requires timing-sensitivity since it explicitly states that an attacker must not learn anything from observing an event $\alpha$ and its timestamp. 
Another notable change is Clause 1a that specifies that the timing of a regular declassification must not convey information.
Finally, this definition also changes the semantics of the $\code{tini}$-construct (cf. Clause 2b).
Instead of declassifying the progress knowledge it now declassifies the timing behavior of the code block guarded by $\code{tini}$.

\section{Discussion}

\label{sec:future-work}

\paragraph{Dimensions and principles of declassification} 
The reframing of the progress-insensitive security as declassification allows us to think about it in terms of declassification principles and dimensions. The locality-driven aspect of our definition places it in a \emph{where} dimension, while the use of authority-based bounds naturally has a clear \emph{what} flavor.  While we do not specify any bounds on what information can be learned via a \code{tini}-declassification as long as the authority is sufficient, 
the prior work on tight specification of \emph{what} information is released through declassifications~\cite{AskarovSabelfeld,chudnov2018assuming} should compose with our definition.
Our authority model is inspired by the expressive label models such as DLM~\cite{dlm} and FLAM~\cite{flam}; and studying our condition in the formal frameworks of these label models will lead to \emph{who} characterizations of the \code{tini}-declassifications.

Another interesting angle to explore is the integration of integrity into the formal model, which would allow one to study the robustness~\cite{robustness} of declassifications via progress-insensitivity. Here, a potentially desirable semantic characterization is that attacker-controlled input does not influence information leaked through termination channels. A knowledge-based approach to robustness~\cite{AskarovMyers} can provide a starting point for such a definition.

With respect to the four principles of declassification, we believe that the principles of semantic consistency -- namely that security definition should be invariant under equivalence-preserving transformations -- and of conservativity -- namely that the definition of security should be a weakening of noninterference -- follow directly from the knowledge-based nature of the definition that is inherently attacker-driven~\cite{bastys2018prudent}. The principle of monotonicity of release -- namely that adding a declassification should not make a secure program insecure -- is also satisfied by our definition: adding a \code{tini} block to a program that is already accepted by Definitions~\ref{def:bpni} does not change how the definition treats this program, because all knowledge containments for the declassification cases are weaker than Clause 3 of the definition (a similar argument applies to the normal declassification). Finally, our definition also satisfies the non-occlusion principle -- namely, that the presence of declassifications should not mask other covert leaks. This one has two subtleties. The first one is already discussed in Section~\ref{twoclauses}. The second one is that without Clause 1a of the definition, we would have violated non-occlusion, as examples that reach an explicit declassification after a high loop would have been accepted. 

Similar arguments apply to Definition~\ref{def:btni}.

\paragraph{Design principle for pc-declassifications}
In the information flow community, pc-declassifications have a poor reputation because their security characterization has been not well understood. Our work provides a principle for understanding security of pc-declassifications that can answer the following question:
\emph{given a programming language or a system that has a primitive for pc-declassification, how dangerous is it?} The key to answering this question is bounding the progress knowledge.

If the security of pc-declassification can be characterized as a bound on progress knowledge -- as we do in Definition~\ref{def:bpni} -- then these pc-declassifications are as dangerous as leaks through progress. However, if progress knowledge cannot be bounded, then these pc-declassifications are more dangerous. For example, in a system designed to allow any pc-declassifications, programs such as
    $$\IfThenElse{h}{\PCDecl{} {\mathit{H}_\mathit{auth}} {L};  \Assign{l}{0}}{\PCDecl{} {\mathit{H}_\mathit{auth}} {L}; \Assign{l}{1}}$$ can  leak information indirectly more efficiently than just encoding the secret in the length of the trace.

\paragraph{Access control to authority} Neither our security policy nor the language provides guarantees about programs that misuse authority if they have access to it. To that extent, our approach leaves it to the programmers to ensure that untrusted code does not have access to authority above the code's intended security clearance. However, the capability-based nature of the authority means that a complementary technique for principled control of capabilities can be used. One candidate approach is the work by \textcite*{Dimoulas} that uses access control and integrity policies to restrict capability use.
Another is the mechanism of  bounded privileges for LIO proposed by \textcite*{ItsMyPrivelege}.

\paragraph{Enforcement techniques} We choose a simple runtime monitor to showcase the enforcement of the new definition. While the monitor is fully dynamic and flow-insensitive, we believe that other single-trace monitoring techniques such as hybrid information flow monitors~\cite{Moore,AskarovMantelChong,JsFlow} as well as Denning-style static techniques can be easily adapted. 
Static approaches may have an added benefit of helping infer the location of \code{tini} statements.
An interesting prospect for future work is extension of monitors designed for declassification for secure multi execution~\cite{willard,limin} to enforce our definition.

\paragraph{Timing treatment}
Our treatment of timing-sensitivity in Section~\ref{sec:timingsensitivity} via a simple step counter is admittedly academic, given the plethora of architectural and runtime side channels today. We nevertheless believe that the formulation of the timing-sensitive condition is useful, and can be combined with other proposals to mitigate practical timing attacks such as predictive mitigation~\cite{askarov2010predictive,zhang2011predictive,zhang2012language}.

\section{Implementation experience}
\label{sec:implementation}
We implemented the \code{tini}-based enforcement as a part of Troupe~\cite{Troupe}.
%
%Since this setting allows sending code between potentially untrusted nodes,
%it is not sufficient to use a progress-insensitive monitor.
%On the other hand, since the language both allows blocking behavior and possible divergence, it is hard to precisely enforce progress-sensitive security in this setting.
%The monitor therefore needs to be very conservative and accrues all taint from branching points in the program.
%In this setting, it is crucial to have a way for the programmer to downgrade the taint accumulated from processing sensitive data.
%
This language enforces progress-sensitive security, but 
allows \code{tini}-scoped initialization as a variation of let-declarations
\begin{lstlisting}[language=troupe]
let tini auth  (* tini declaration *)
    val v1 = e1
    val v2 = e2
    ...
in (* the point of pcdecl *)
    e
end 
\end{lstlisting}
This construct declassifies the termination of the initialization expressions \code{e1, e2,...}
using authority \code{auth} before evaluating the body~\code{e}.

%\subsection{News widget example}
Figure~\ref{fig:news_widget:troupe} presents a
snippet from the code of the news widget 
example in our language.
The top listing is the source of the news widget itself. 
When invoked with the favorite topic and its current state as arguments,
it updates the counter, fetching
updated news from the remote servers if necessary. Finally, it returns the 
result together with the updated state. Fetching the news is potentially blocking and
 implemented in the function \code{fetch\_news} (omitted from the listing but it uses the networking primitives of the language). The \code{news} value is an associative list, and the secret-dependent lookup
is done using the built-in function \code{list\_lookup\_with\_default}.
The initial state of the widget is an empty list, with the counter set to zero. The security
level of the initial state is \code{NEWS}.

The bottom listing in the figure displays how 
this widget is used by user at level \code{ALICE}. The important part is the invocation of the 
\code{news\_widget} is placed in the \code{let\ tini} block with attenuated authority \code{NEWS}, which 
limits the termination leakage of the \code{news\_widget} function.

\begin{figure}
\begin{lstlisting}[language=troupe]
fun news_widget fav state =
  let 
    val (news, update_counter) = state
    val news = if update_counter %10 = 0
                 then fetch_news() (* Blocking *)
                 else news
    val update_counter = update_counter + 1
    (* Operation on the secret *)
    val fav_news = list_lookup_with_default 
                           news fav "no news"
  in (fav_news, (news,update_counter))
  end
val init_state = ([], 0) raisedTo {NEWS}
\end{lstlisting}
\begin{lstlisting}[language=troupe]
(* Receiving widget and initial state *)
val (news_widget, state0) = fetch_widget ()

(* Usage of the widget by user ALICE *)
val news_auth = attenuate(rootauth, {NEWS})
val fav_topic = "#politics" raisedTo {ALICE}

(* Calling untrusted widget code *)
val (fav_news1, state1) =
       let tini news_auth
         val res = news_widget fav_topic state0
       in res 
       end
\end{lstlisting}
\caption{News widget (top) and its usage (bottom) code snippets\label{fig:news_widget:troupe}}
\end{figure}

The actual example is about 80 lines of code. As another data point for the readers, a different case study
in our language of roughly 500LOC uses the \code{let\ tini} construct 9 times.

\section{Related work}
\label{sec:related-work}

\paragraph{pc-declassification}
Jif provides a mechanism for pc-downgrading in the form of a declassify statement that lowers that \code{pc}-label
that is tracked by the type system. Unlike other features of Jif that are proven sound, e.g., dependent labels~\cite{DependentLabels} or robust declassification~\cite{ChongMyers}, there is no soundness theorem for the pc-declassifications.

Both the Asbestos~\cite{asbestos} and the HiStar~\cite{Histar} operating systems also allow downgrading of the control-flow.
In Asbestos a process with \textit{privilege}, a related notion to our authority, can \textit{decontaminate} other processes' \textit{send} label which
has the effect of allowing the other process to ``forget'' that it has previously seen secret data from the privileged process.
In our setup, this corresponds to passing an authority that allows declassifying control-flow up to the senders level.
HiStar similarly makes it possible to lower the accrued taint by passing on untainting gates that act as a capability for lowering the $\pc$.
Both of these systems provide this functionality because it is a practical feature to have,
but neither of them presents a security condition that encapsulates what this feature entails regarding leakage.

Chandra and Franz~\cite{info-flow-java-vm} present an information flow framework for the Java Virtual Machine with a hybrid monitoring that uses a static analysis to reason about when it is safe to declassify the $\pc$.
Similarly to earlier work by, for example, Denning~\cite{Denning:1977:CPS:359636.359712}, they statically find the immediate postdominator (the nearest join-point that all execution paths must pass through) to any branch-point and insert a $\pc$-lowering command at this point.
Their security condition is intended to only allow lowering the $\pc$ when no knowledge is revealed by doing so,
but since they are in a setting where almost any bytecode can throw unchecked exceptions, this is not generally feasible.
Instead, they disregard all implicit flows through unchecked exceptions and accept these leaks as a limitation of the security the system provides.
We believe one could extend this line of work by applying our bound on what is learned through such flows,
and thereby gain a stronger guarantee for the system as a whole.

The idea of control flow declassification also appears in the discussions of information flow control vs. 
taint tracking. For example, \textcite*{ExplicitSecrecy} use an observational approach where every branch decision is declassified.

\paragraph{Knowledge-based policies}
The methodology and the experience of this paper is in line with the argument by \textcite*{broberg2015anatomy} that epistemic specifications is the most natural way to specify information flow properties: 

\begin{displayquote}
The notion of security intrinsically has nothing to do with observing two separate runs – but rather what can be deduced from observing a single run. [\dots] A two-run formulation could certainly be very useful as part of the strategy to prove e.g. the correctness of an enforcement mechanism. [\dots] But that property is then only a stepping stone, and should, for completeness, be shown to imply the natural epistemic property.
\end{displayquote}
In our case, it is the operational security (cf. Lemma~\ref{lem:NI-bridge}) that has the two-run formulation.

The knowledge-based approach we use in this work follows the style of definitions of gradual release~\cite{Gradualrelease}. 
Logical epistemic approaches include the work by \textcite{halpern2002secrecy} that use epistemic logic to specify noninterference, and that of \textcite*{balliu2011epistemic} that uses epistemic temporal logic used to reason about knowledge acquired by observing program outputs.

\textcite{chudnov2018assuming} define an epistemic semantics for relational assumptions and guarantees in a progress-insensitive setting. To specify the allowed knowledge at a particular point in the trace they define a notion of release policy of a trace, where relational assumptions are interpreted as an annotation permitting the attacker to learn new information. The insight of our work suggests the direction of lifting their approach into a progress-sensitive setting and treating progress leaks as another form of relational assumptions.

McCall et al.~\cite{limin} propose a model for enforcing information flow control in the setting of webpages that must handle execution of untrusted scripts.
Their approach enforces robust declassification such that untrusted code cannot influence what is declassified by extending prior work on secure-multi-execution.
They show their enforcement sound with respect to a knowledge-based progress-insensitive noninterference condition with declassification.
They also present a progress-sensitive notion of noninterference, but restrict their focus to the weaker progress-insensitive condition, because IO-operations can use potentially looping event handlers that leak information through progress (a design decision somewhat reminiscent of the scenario in the Introduction). 
In the context of their work, the bridge between progress-sensitive and progress-insensitive security provided by our definition, can allow programmers to explicitly state when, and how much, information an event handler is allowed to leak through divergence.

\paragraph{Leakage via termination}

\citeauthor{moore2012precise}~\cite{moore2012precise} propose a type-based enforcement combined with a runtime mechanism for budgeting the amount of information leaked through termination at runtime. The idea is to use a termination oracle that uses maximum available runtime public information to deduce the termination behavior of secret-dependent code. The budgets mechanism allows for a quantitative interpretation of the leakage.

\paragraph{Untrusted code}
LIO~\cite{LIO}, MAC~\cite{mac}, and related programming models side step the issue of label creep via a programming discipline where 
high computations are forked into separate processes. A consequence of this programming model however is that consuming the result of the forked computation requires process synchronization followed by 
explicit declassification.
Fabric~\cite{liu2017fabric} contains a number of mechanisms for confining untrusted code downloaded over a network, including limits on authority that the code can use and access labels that limit when the untrusted code can read remote objects. As Fabric is based on Jif, it also places timing and progress channels outside of its threat model.

\section{Conclusion}
\label{sec:conclusion}
This paper proposes two novel knowledge-based security conditions that capture the semantic meaning of declassifying the progress knowledge in information flow control systems.
While many language-based and architectural systems allows such declassification there is, to the best of our knowledge, no formal characterization of it.
We present a language construct, $\code{tini}$, that exactly captures the embedding of progress-insensitive code in a stricter setting and show how this can be used in the presence of potentially blocking or diverging untrusted code.
We furthermore show that our conditions are enforceable by a mostly standard dynamic monitor.
For future work we conjecture that our epistemic definitions can form a foundation for further studies by extending it with
for example integrity and robust downgrades, principled usage of authority-capabilities, or more elaborate label models.
%\
Finally, we believe that a large body of techniques that rely on progress-insensitive security can use the insight of our work to accommodate stronger adversary models.

\section{Acknowledgements}
We thank Mathias Vorreiter Pedersen for his help with the technical aspects of this work at 
an earlier stage, 
Alix Trieu and Andrei Sabelfeld for their comments and insights, and the anonymous reviewers for their suggestions for improving the presentation of this paper.
 This work is supported by the DFF project 6108-00363 from The Danish Council for Independent Research for the Natural Sciences (FNU) and Aarhus University Research Foundation.

% \clearpage

\printbibliography

\opt{extended}{% \makeatletter
% \renewenvironment{description}%
%                {\list{}{\leftmargin=5pt % <------- Adjust this length
%                         \labelwidth\z@ \itemindent-\leftmargin
%                         \let\makelabel\descriptionlabel}}%
%                {\endlist}
% \makeatother
\onecolumn
\pagebreak
\setlist[description]{style=nextline,leftmargin=1cm}
\appendix
The rest of this document serves to prove \fullref{thm:soundness-bpni}.
To do so, we first provide a few auxiliary definitions and lemmas leading up to the proofs.

\subsection{Well formed expressions}
We restrict the occurrence of $\PCDecl{\eta}{\ell_1}{\ell_2}$ such that the usage is well-bracketed w.r.t the operational semantics:
\begin{figure}[h]
  \begin{mathpar}
    \inferrule
    {
      ~
    }
    {
      \declWF{\PCDecl{\eta}{\ell_1}{\ell_2}}
    }
    \and
    \inferrule
    {
      \cmd \not = \cmd_1;\cmd_2\\
      \pcdeclfree{\cmd}
    }
    {
      \declWF{\cmd}
    }
    \and
    \inferrule
    {
      \declWF{\cmd_1}\\
      \declWF{\cmd_2}
    }
    {
      \declWF{\cmd_1;\cmd_2}
    }
  \end{mathpar}
\end{figure}

\begin{lemma}[Well formedness is preserved by the semantics]\label{lem:wf-preserved}
  for any command $c$, if $\declWF{c}$ and
  $\configThree{c}{m}{\pc} \stepsmanywith{\events} \configThree{c'}{m'}{\pc'}$
  then $\declWF{c'}$ holds.
\end{lemma}
\begin{proof}
  Immediate by induction on $c$.
\end{proof}

\subsection{Indistinguishability relations}
% NOTE: already in main paper:
% We consider equivalence relations $\indmeta$ on memories. We write
% $\indmeta \subseteq \loweq{\ell}$ when $\indmeta$ is
% smaller than low-equivalence at some level $\ell$.

\begin{definition}[Indistinguishability propagation by bridge]
  Given an indistinguishability relation $\indmeta$, define
  \emph{indistinguishability propagation} from configuration
  with command $c$ and pc register $\pc$, denoted
  $\indprop{\indmeta}{c}{\pc}{\ell}{\alpha_1 \dots \alpha_n}$ to be the relation
  such that
  \begin{mathpar}
    \inferrule{
      m \mathrel{I} s
    }{
      m \indprop{\indmeta}{c}{\pc}{\ell}{\emptylist} s
    }
    \and
    \inferrule{
      m \mathrel{I} s
      \and
      \syncConfig{c}{m}{s}{\pc} \syncbridge{\ell}{\alpha_1 \dots \alpha_k}
      \syncConfig{c'}{m'}{s'}{\pc''}
    }{
      m' \indprop{\indmeta}{c}{\pc}{\ell}{\alpha_1 \dots \alpha_k} s'
    }
  \end{mathpar}
\end{definition}

\begin{lemma}[Preservation of $\ell$-equivalence by bridge propagation]
  \label{lemma:propagation-ell-equivalence}
  If $\indmeta \subseteq \loweq{\ell}$ then
  $\indprop{\indmeta}{c}{\pc}{\ell}{\alpha_1 \dots \alpha_k} \subseteq \loweq{\ell}$.
\end{lemma}
\begin{proof}By induction on $k$. We examine the inductive case, as the base case is straightforward.

  Consider $m', s'$ such that $m' \mathrel{\indprop{\indmeta}{c}{\pc}{\ell}{\alpha} \subseteq \loweq{\ell}} s'$.
  Unfolding the definitions, it must be that there are
  $m$ and $s$ such that $m \loweq{\ell} s$ and $m \mathrel{I} s$.
  Since the bridge relations update the memories with the same
  $\ell$-equivalent events, then it must be that $m' \loweq{\ell} s'$.
\end{proof}

\begin{lemma}[Restriction monotonicity]
  $\indnarr{\indmeta}{c}{\pc}{\ell}{\eventMeta} \subseteq \indmeta$.
\end{lemma}
\begin{proof}
  Immediate from the definition of $\indnarr{\cdot}{c}{\pc}{\ell}{\eventMeta}$
\end{proof}

\begin{lemma}[Sequence decomposition]
  \label{seq_decomp}
  Suppose
  $\configThree{c_1;c_2}{m}{\pc} \lbridgestep{n}{\ell}{\alpha} \configThree{c'}{m'}{\pc'}$,
  then one of the following holds
  \begin{enumerate}
  \item $\configThree{c_1}{m}{\pc} \lbridgestep{n}{\ell}{\alpha} \configThree{c_1'}{m'}{\pc'}$
    and $c' = c_1';c_2$
  \item $\configThree{c_1}{m}{\pc} \lbridgestep{n_1}{\ell}{\beta} \configThree{\Stop}{m_1}{\pc_1}$ and
    $\configThree{c_2}{m_1}{\pc_1} \lbridgestep{n_2}{\ell}{\alpha} \configThree{c'}{m'}{\pc'}$ where $n = n_1 + n_2 + 1$ and $\nobs{\beta}{\ell}$.
  \end{enumerate}
\end{lemma}
\begin{proof}
  By inspection of the rules in the bridge relation and the associated rules of the operational semantics.
\end{proof}

\begin{lemma}[Equivalent runs are synchronized]
  \label{lem:runs-to-sync-bridge}
  If we have runs
  $\configThree{c}{m}{\pc} \stepsmanywith{\events} \configThree{c_1}{m_1}{\pc_1}$ and
  $\configThree{c}{s}{\pc} \stepsmanywith{\events'} \configThree{c_2}{s_2}{\pc_2}$,
  where the initial memories satisfy that
  $m \loweq{\level} s$
  and their traces are equal up to some level $\level$, $\tproj{\events}{\level} = \tproj{\events'}{\level}$, then there exists $c'$, $m'$, $s'$, and $\pc'$ such that
  $$\syncConfig{c}{m}{s}{\pc} \syncbridge{\ell}{\alpha_1, \alpha_2, \ldots, \alpha_k} \syncConfig{c'}{m'}{s'}{\pc'}$$
  where
  $\tproj{\events}{\level} = \tproj{\events'}{\level} = [\alpha_1, \alpha_2, \ldots, \alpha_k]$
  and $m' \loweq{\level} m_1$ and $s' \loweq{\level} s_1$.
\end{lemma}
\begin{proof}
  Follows from determinism of the operational semantics and the fact that observable events capture all observable changes to memories.
\end{proof}

\begin{lemma}[Noninterference of expressions]\label{lem:NI-expressions}
  Given an expression $\expr$ and two memories $\mem_1$ and $\mem_2$ such that $\mem_1 \loweq{\level} \mem_2$.
  If
  $\configTwo{\expr}{\mem_1} \yields \LabeledValue{\basemeta_1}{\ell_1}$
  and
  $\configTwo{\expr}{\mem_2} \yields \LabeledValue{\basemeta_2}{\ell_2}$
  then
  $\typeof{\basemeta_1} = \typeof{\basemeta_2}$
  and
  $(\ell_1 \sqsubseteq \ell \lor \ell_2 \sqsubseteq \ell) \implies (\basemeta_1 = \basemeta_2 \land \ell_1 = \ell_2)$.
\end{lemma}
\begin{proof}
  Straightforward induction on the evaluation rules.
\end{proof}

\begin{lemma}[Barring $\code{pcdecl}$ commands, $\pc$ never decreases]\label{lem:pc-increases-pcdeclfree}
  Given a bridge-step $\configThree{c}{m}{\pc} \lbridgestep{n}{\ell}{\alpha} \configThree{c'}{m'}{\pc'}$ where $\pc \not \sqsubseteq \ell$ and $\pcdeclfree{c}$
  then it holds that $\pc' \not \sqsubseteq \ell$
\end{lemma}
\begin{proof}
  Straightforward induction on the bridge-step relation.
\end{proof}

\begin{lemma}[Observable (non-$\EventTiniExit{}{}{}{}$)-events are only emitted in low contexts]\label{lem:obs-event-implies-low-pc}
  Given a bridge-step $\configThree{c}{m}{\pc} \lbridgestep{n}{\ell}{\alpha} \configThree{c'}{m'}{\pc'}$ where $\pc' \sqsubseteq \ell$ and $\alpha$ is not an observable $\EventTiniExit{}{}{}{}$-event,
  it holds that $\pc \sqsubseteq \ell$.
\end{lemma}
\begin{proof}
  Straightforward induction on the bridge-step relation.
\end{proof}

\subsection{Proof of operational definition}
\textit{Proof of \autoref{lem:NI-bridge}:}
Given
$\configThree{c}{m}{\pc} \lbridgestep{n}{\ell_\adv}{\alpha} \configThree{c'}{m'}{\pc'}$,
we proceed by strong induction in $n$.
\begin{description}
\item[For $n = 0$:]
  We have that $\configThree{c}{m}{\pc} \lbridgestep{0}{\ell_\adv}{\alpha} \configThree{c'}{m'}{\pc'}$
  which, by inversion, entails that we must either have that $c' =\Stop$ and $\alpha \not \flowsto \ell_\adv$ or that $\alpha \flowsto \ell_\adv$.
  \begin{description}
  \item[$\alpha \flowsto \ell_\adv$:]
    We have the following cases for $\alpha$:
    \begin{description}
    \item [Case $\alpha = \EventAssign{x}{v}$:]
      Since $c$ emits $\EventAssign{x}{v}$ in a single evaluation step it must be the case that
      $c$ is an assignment $\Assign{x}{e}$ where $\configTwo{e}{m} \yields v$.
      From \fullref{lem:NI-expressions} we have that
      $\configTwo{e}{s} \yields v$
      so it follows trivially that
      $\configThree{c}{s}{\pc} \lbridgestep{0}{\ell_\adv}{\EventAssign{x}{v}} \configThree{c'}{s'}{\pc'}$
      where $m' \loweq{\level_\adv} s'$ which is what we need to prove.
    \item [Case $\alpha = \EventDeclassify{x}{\ell_{\mathit{auth}}}{\ell_{\mathit{from}}}{\ell_{\mathit{to}}}$:]
      Since $c$ emits $\EventDeclassify{x}{\ell_{\mathit{auth}}}{\ell_{\mathit{from}}}{\ell_{\mathit{to}}}$ in a single evaluation step it must be the case that
      $c$ is a declassify command, $\Declassify{x}{e_{\mathit{auth}}}{\level_{\mathit{to}}}{e_v}$,
      $\configTwo{e_v}{m} \yields \LabeledValue{v}{\ell_{\mathit{from}}}$
      and
      $\configTwo{e_{\mathit{auth}}}{m} \yields \LabeledValue{\AuthorityVal{\level_{\mathit{auth}}}{1}}{\level'}$
      .
      Furthermore, it must be the case that $\level_{\mathit{from}} \flowsto \level_{\mathit{auth}} \sqcup \level_{\adv}$ and
      $\level' \flowsto \pc$.

      We have two cases to show:
      \begin{description}
      \item[Case~\ref{bridge-item:1}:]
        Suppose we have
        $s$ such that $s \loweq{\level_{\mathit{auth}} \sqcup \level_{\adv}} m$.
        We need to show that
        $\configThree{\Declassify{x}{e_{\mathit{auth}}}{\level_{\mathit{to}}}{e_v}}{s}{\pc}
        \lbridgestep{0}{\ell_\adv}{\EventDeclassify{x}{\ell_{\mathit{auth}}}{\ell_{\mathit{from}}}{\ell_{\mathit{to}}}}
        \configThree{c'}{s'}{\pc'}$ and $s' \loweq{\level_\adv} m'$
        which follows from applying \fullref{lem:NI-expressions} on the evaluations of $e_v$ and $e_a$.
      \item[Case~\ref{bridge-item:3}:]
        Suppose we have $s$ such that $s \loweq{\level_\adv} m$.
        We need to show that 
        $\configThree{\Declassify{x}{e_{\mathit{auth}}}{\level_{\mathit{to}}}{e_v}}{s}{\pc}
        \lbridgestep{0}{\ell_\adv}{\EventDeclassify{x}{\ell_{\mathit{auth}}}{\ell_{\mathit{from}}}{\ell_t}}
        \configThree{c'}{s'}{\pc'}$.
        
        By \fullref{lem:NI-expressions} we have that
        $\configTwo{e_{\mathit{auth}}}{s} \yields \LabeledValue{\AuthorityVal{\level_{\mathit{auth}}}{1}}{\level'}$
        and by completeness of expression evaluation we have that
        $\configTwo{e_v}{m} \yields \LabeledValue{v'}{\ell_{\mathit{from}}}$
        so what we need follows directly from the semantics of the language.
      \end{description}
    \item [Case $\alpha = \EventTiniExit{\eta}{\ell_{\mathit{auth}}}{\ell_{\mathit{from}}}{\ell_{\mathit{to}}}$:]
      Command $c$ must only consist of a $\PCDecl{\eta}{\ell_{\mathit{auth}}}{{\level_{\mathit{to}}}}$ and $\pc \flowsto \level_{\mathit{auth}} \sqcup \level_\adv$
      so it follows trivially that for any
      $s$ such that $s \loweq{\level_{\mathit{auth}} \sqcup \level_\adv} m$ we have that
      $\configThree{ \PCDecl{\eta}{\ell_{\mathit{auth}}}{\ell_{\mathit{to}}}}{s}{\pc}
      \lbridgestep{0}{\ell_\adv}{\EventTiniExit{\eta}{\ell_{\mathit{auth}}}{\ell_{\mathit{from}}}{\ell_{\mathit{to}}}}
      \configThree{c'}{s'}{\pc'}$ and $s' \loweq{\level_\adv} m'$.
    \end{description}
  \item[$c' = \Stop$ and $\alpha \not \flowsto \ell_\adv$:]
    We have that $c' = \Stop$ and $\nobs\alpha{\level_\adv}$.
    We have two cases based on whether or not $\pc' \flowsto \level_\adv$.
    \begin{description}
    \item[Case $\pc' \flowsto \level_\adv$:]
      It must also be the case that $\pc \flowsto \level_\adv$,
      and it therefore easily follows that if $s \loweq{\level_\adv} m$
      we also have that $\configThree{c}{s}{\pc} \lbridgestep{0}{\level_\adv}{\alpha} \configThree{\Stop}{s'}{\pc'}$, where $s' \loweq{\level_\adv} m'$.
    \item[Case $\pc' \not \flowsto \level_\adv$:] It must be that
      $\configThree{c}{s}{\pc} \lbridgestep{0}{\level_\adv}{\alpha'} \configThree{c''}{s'}{\pc''}$.
      But since $\alpha$ is not observable neither is $\alpha'$ so $c'' = \Stop$.
      Similarly, $\alpha'$ cannot be emitted by a command that changed $\level_\adv$-parts of memory, so it also holds that $s' \loweq{\level_\adv} m'$.
      Finally, by examining the determinism of the operational semantics for $\level_\adv$-equivalent memories we get that $\pc'' \not \flowsto \level_\adv$.
    \end{description}
  \end{description}
\item[For $n = k + 1$:]
  We have
  $\configThree{c}{m}{\pc} \lbridgestep{k+1}{\ell_\adv}{\alpha} \configThree{c'}{m'}{\pc'}$.
  % which must be the result of applying \textsc{Bridge-Multi} so there must exist $c_1$, $m_1$ and $\pc_1$ such that
  % $\configThree{c}{m}{\pc} \stepswith{\beta} \configThree{c_1}{m_1}{\pc_1}$ where $\nobs{\beta}{\level_\adv}$ and
  % $\configThree{c_1}{m_1}{\pc_1} \lbridgestep{k}{\ell_\adv}{\alpha} \configThree{c_1}{m_1}{\pc_1}$.
  We proceed by induction in $c$.
  \begin{description}
  \item[$c$ is an assignment, declassify, or skip:]
    In all cases we get a contradiction since they all yield $\Stop$ in a single step which would mean that $k + 1 = 0$.
  \item[$c$ is sequence $d_1;d_2$:]
    We have four cases to prove:
    \begin{description}[style=nextline]
    \item [If $\alpha = \EventDeclassify{x}{\level_{\mathit{from}}}{\level_{\mathit{auth}}}{\level_{\mathit{to}}}$ and $\obs{\alpha}{\level_\adv}$:]
      We need to show Case~\ref{bridge-item:1} and Case~\ref{bridge-item:3}.
      \begin{description}
      \item[Proof of Case~\ref{bridge-item:1}:]
        Let
        $\indmeta = \indnarr{\loweq{\level_{\mathit{auth}} \sqcup \level_{\adv}}}{c}{\pc}{\level_{\mathit{auth}} \sqcup \level_{\adv}}{\beta_1 ,\ldots, \beta_j}$
        where
        \begin{multline*}
          \configThree{d_1;d_2}{m}{\pc} \lbridgestep{i_1}{\level_{\mathit{auth}} \sqcup \level_{\adv}}{\beta_1} \configThree{c_1}{m_1}{\pc_1} 
          \lbridgestep{i_2}{\level_{\mathit{auth}} \sqcup \level_{\adv}}{\beta_2} \ldots \lbridgestep{i_j}{\level_{\mathit{auth}} \sqcup
            \level_{\adv}}{\beta_j}
          \configThree{c_j}{m_j}{\pc_j}\\
          \lbridgestep{i'}{\level_{\mathit{auth}} \sqcup \level_{\adv}}{\EventDeclassify{x}{\level_{\mathit{from}}}{\level_{\mathit{auth}}}{\level_{\mathit{to}}}} \configThree{c'}{m'}{\pc'}
        \end{multline*}
        such that $i_1 + \ldots + i_j + j + i' = k + 1$
        and $s$ such that $m \mathrel{\indmeta} s$ be given.
        We now have two cases based on $j$:
        \begin{description}
        \item[Case $j$ is $0$:]
          Since there are no events observable at level $\level_{\mathit{auth}} \sqcup \level_{\adv}$ we have that
          $\configThree{d_1;d_2}{m}{\pc} \lbridgestep{k + 1}{\level_{\mathit{auth}} \sqcup \level_{\adv}}{\EventDeclassify{x}{\level_{\mathit{from}}}{\level_{\mathit{auth}}}{\level_{\mathit{to}}}} \configThree{c'}{m'}{\pc'}$.

          By applying \fullref{seq_decomp} we have two cases:
          \begin{description}
          \item[$\alpha$ is produced by $d_1$]
            We have
            $\configThree{d_1}{m}{\pc} \lbridgestep{k + 1}{\level_{\mathit{auth}} \sqcup \level_{\adv}}{\EventDeclassify{x}{\level_{\mathit{from}}}{\level_{\mathit{auth}}}{\level_{\mathit{to}}}} \configThree{d_1'}{m'}{\pc'}$ and $c' = d_1' ; d_2$.
            We directly get what we need by applying the inner induction hypothesis to this run.
          \item[$\alpha$ is produced by $d_2$]
            We have
            $$\configThree{d_1}{m}{\pc} \lbridgestep{k_1}{\level_{\mathit{auth}} \sqcup \level_{\adv}}{\beta} \configThree{\Stop}{m_1}{\pc_1}$$
            and
            $$\configThree{d_2}{m_1}{\pc_1} \lbridgestep{k_2}{\level_{\mathit{auth}} \sqcup \level_{\adv}}{\EventDeclassify{x}{\level_{\mathit{from}}}{\level_{\mathit{auth}}}{\level_{\mathit{to}}}} \configThree{c'}{m'}{\pc'}$$
            where $\nobs{\beta}{\level_{\mathit{auth}} \sqcup \level_{\adv}}$ and $k_1 + k_2 = k$.
            It follows from applying \fullref{lem:obs-event-implies-low-pc} to the run for $d_2$ that $\pc_1 \sqsubseteq \level_{\mathit{auth}} \sqcup \level_{\adv}$
            so we can apply the induction hypothesis (with the attacker level instantiated to ${\level_{\mathit{auth}} \sqcup \level_{\adv}}$) on the run for $d_1$ and obtain from Case~\ref{bridge-item:3} that
            $$\configThree{d_1}{s}{\pc} \lbridgestep{k_1}{\level_{\mathit{auth}} \sqcup \level_{\adv}}{\beta'} \configThree{\Stop}{s_1}{\pc_1}$$
            where
            $m_1 \loweq{\level_{\mathit{auth}} \sqcup \level_{\adv}} s_1$
            and using this we can apply the induction hypothesis to the $d_2$-run and obtain what we need from Case~\ref{bridge-item:1}.
          \end{description}
        \item[Case $j > 0 $:]
          From the definition of
          $\indnarr{\loweq{\level_{\mathit{auth}} \sqcup \level_{\adv}}}{c}{\pc}{\level_{\mathit{auth}} \sqcup \level_{\adv}}{\beta_1 ,\ldots, \beta_j}$
          we have that
          $\syncConfig{c}{m}{s}{\pc} \syncbridge{\level_{\mathit{auth}} \sqcup \level_{\adv}}{\beta_1 \dots \beta_j} \syncConfig{c_j}{m_j}{s_j}{\pc_j}$,
          and it therefore follows from \fullref{lemma:propagation-ell-equivalence} that $m_j \loweq{\level_{\mathit{auth}} \sqcup \level_{\adv}} s_j$.
          It must be the case that $i' < k+1$ and we can therefore apply the induction hypothesis on
          $\configThree{c_j}{m_j}{\pc_j} \lbridgestep{i'}{\level_{\mathit{auth}} \sqcup \level_{\adv}}{\alpha} \configThree{c'}{m'}{\pc'}$ to obtain what we need.
        \end{description}

      \item[Proof of Case~\ref{bridge-item:3}:]
        Suppose we are given $s$ where $m \loweq{\level_\adv} s$.
        We need to show that
        $$
        \configThree{d_1;d_2}{s}{\pc} \lbridgestep{n'}{\level_\adv}{\EventDeclassify{x}{\level_{\mathit{from}}}{\level_{\mathit{auth}}}{\level_{\mathit{to}}}}\configThree{c'}{s'}{\pc'}
        $$

        By applying \fullref{seq_decomp} we have two cases based on whether or not the event is produced by the first or second part of the sequential composition:
        \begin{description}
        \item[$\alpha$ is produced by $d_1$] Then what we need follows directly from applying the inner induction hypothesis.
        \item[$\alpha$ is produced by $d_2$]
          We have that

          $$\configThree{d_1}{m}{\pc} \lbridgestep{k_1}{\ell_\adv}{\beta} \configThree{\Stop}{m_1}{\pc_1}$$
          and
          $$\configThree{d_2}{m_1}{\pc_1} \lbridgestep{k_2}{\ell_\adv}{\alpha} \configThree{c'}{m'}{\pc'}$$
          where $\nobs{\beta}{\level_\adv}$.
          Since we end with a low program-counter, $\pc' \sqsubseteq \level_\adv$, it follows from applying \fullref{lem:obs-event-implies-low-pc} to the run for $d_2$ that $\pc_1 \sqsubseteq \level_\adv$.
          Therefore, by applying the inner induction hypothesis on
          $\configThree{d_1}{m}{\pc} \lbridgestep{k_1}{\ell_\adv}{\beta} \configThree{\Stop}{m_1}{\pc_1}$
          and from Case~\ref{bridge-item:3} we obtain run
          $\configThree{d_1}{s}{\pc} \lbridgestep{k_1'}{\ell_\adv}{\beta} \configThree{\Stop}{s_1}{\pc_1}$
          such that
          $m_1 \loweq{\level_\adv} s_1$.
          We can now apply the inner induction hypothesis on
          $\configThree{d_2}{m_1}{\pc_1} \lbridgestep{k_2}{\ell_\adv}{\alpha} \configThree{c'}{m'}{\pc'}$
          such that Case~\ref{bridge-item:1} gives us that
          $\configThree{d_2}{s_1}{\pc_1} \lbridgestep{k_1'}{\ell_\adv}{\alpha} \configThree{c'}{s'}{\pc'}$
          and $m' \loweq{\level_\adv} s'$
          which is exactly what we need.
        \end{description}
      \end{description}
    \item [If $\alpha = \EventTiniExit{\eta}{\level_{\mathit{auth}}}{\level_{\mathit{auth}} \sqcup \level_{\mathit{to}}}{\level_{\mathit{to}}}$ and $\obs{\alpha}{\level_\adv}$:]
      Let
      $\indmeta = \indnarr{\loweq{\level_{\mathit{auth}} \sqcup \level_{\adv}}}{c}{\pc}{\level_{\mathit{auth}} \sqcup \level_{\adv}}{\beta_1 ,\ldots, \beta_j}$
      where
      \begin{multline*}
        \configThree{c}{m}{\pc} \lbridgestep{i_1}{\level_{\mathit{auth}} \sqcup \level_{\adv}}{\beta_1} \configThree{c_1}{m_1}{\pc_1}
        \lbridgestep{i_2}{\level_{\mathit{auth}} \sqcup \level_{\adv}}{\beta_2} \ldots
        \lbridgestep{i_j}{\level_{\mathit{auth}} \sqcup \level_{\adv}}{\beta_j} \configThree{c_j}{m_j}{\pc_j}
      \end{multline*}
      and $s$ such that $m \mathrel{\indmeta} s$ be given.
      Similarly to the case above, we case on $j$; the number of intermediate events that have become observable at the higher attacker-level:
      \begin{description}
      \item[Case $j = 0$:]
        We have that
        $\configThree{d_1; d_2}{m}{\pc} \lbridgestep{k+1}{\level_{\mathit{auth}} \sqcup \level_{\adv}}{\EventTiniExit{\eta}{\level_{\mathit{auth}}}{\level_{\mathit{from}}}{\pc'}}\configThree{c'}{m'}{\pc'}$
        By applying \fullref{seq_decomp} we have two cases:
        \begin{description}
        \item[$\alpha$ is produced by $d_1$]
          We have that
          $$\configThree{d_1}{m}{\pc} \lbridgestep{k + 1}{\level_{\mathit{auth}} \sqcup \level_{\adv}}{\beta} \configThree{d_1'}{m'}{\pc'}$$
          and $c' = d_1'; d_2$, so we are done by aplying the inner induction hypothesis to this run.
        \item[$\alpha$ is produced by $d_2$]
          We have
          $$\configThree{d_1}{m}{\pc} \lbridgestep{k_1}{\level_{\mathit{auth}} \sqcup \level_{\adv}}{\beta} \configThree{\Stop}{m_1}{\pc_1}$$
          and
          $$\configThree{d_2}{m_1}{\pc_1} \lbridgestep{k_2}{\level_{\mathit{auth}} \sqcup \level_{\adv}}{\EventTiniExit{\eta}{\level_{\mathit{auth}}}{\level_{\mathit{from}}}{\pc'}} \configThree{c'}{m'}{\pc'}$$
          where $\nobs{\beta}{\level_{\mathit{auth}} \sqcup \level_{\adv}}$ and $k_1 + k_2 = k$.

          Since the $\code{pcdecl}$ that emits the
          ${\EventTiniExit{\eta}{\level_{\mathit{auth}}}{\level_{\mathit{from}}}{\pc'}}$-event
          is reached by the run in $m$ we know that the $\pc$ must have satisfied that $\pc \sqsubseteq \level_{\mathit{auth}} \sqcup \level_{\adv}$.
          Now since the $\pc$-label cannot decrease below $\level_{\mathit{auth}} \sqcup \level_\adv$ without emitting $\level_{\mathit{auth}} \sqcup \level_\adv$-observable events,
          it must also hold that $\pc_1 \sqsubseteq \level_{\mathit{auth}} \sqcup \level_{\adv}$.

          Furthermore, we know that $\beta$ cannot be a $\code{pcdecl}$ event $\EventTiniExit{\_}{\_}{\_}{\_}$ because then it would have been an observable event at level $\level_{\mathit{auth}} \sqcup \level_{\adv}$ which contradicts $j = 0$. 
          Hence, Case~\ref{bridge-item:3} of the induction hypothesis on the run for $d_1$ applies, and we therefore obtain that
          $$\configThree{d_1}{s}{\pc} \lbridgestep{k_1'}{\level_{\mathit{auth}} \sqcup \level_{\adv}}{\beta'} \configThree{\Stop}{s_1}{\pc_1}$$
          where $s_1 \loweq{\level_{\mathit{auth}} \sqcup \level_{\adv}} m_1$.
          This enables us to apply the induction hypothesis on run for $d_2$ and from Case~\ref{bridge-item:2} we obtain
          $$\configThree{d_2}{s_1}{\pc_1} \lbridgestep{k_2'}{\level_{\mathit{auth}} \sqcup \level_{\adv}}{\alpha'} \configThree{c'}{s'}{\pc'}$$
          where $s' \loweq{\level_{\mathit{auth}} \sqcup \level_{\adv}} m'$.
          It directly follows that $s' \loweq{\level_\adv} m'$ as well and furthermore we can combine the two runs above to obtain
          $$\configThree{d_1;d_2}{s}{\pc} \lbridgestep{k'}{\level_\adv}{\alpha'} \configThree{c'}{s'}{\pc'}$$
          as needed.
        \end{description}
      \item[Case $j > 0$:]
        From the definition of $\indnarr{\loweq{\level_{\mathit{auth}} \sqcup \level_{\adv}}}{c}{\pc}{\level_{\mathit{auth}} \sqcup \level_{\adv}}{\beta_1 ,\ldots, \beta_j}$
        we have that there exists
        $$ \syncConfig{c}{m}{s}{\pc} \syncbridge{\level_{\mathit{auth}} \sqcup \level_{\adv}}{\beta_1 \dots \beta_j} \syncConfig{c_j}{m_j}{s_j}{\pc_j} $$
        and furthermore we have that
        $$ \configThree{c_j}{m_j}{\pc_j} \lbridgestep{i'}{\level_{\mathit{auth}} \sqcup \level_{\adv}}{\EventTiniExit{\eta}{\level_{\mathit{auth}}}{\level_{\mathit{from}}}{\level_{\mathit{to}}}} \configThree{c'}{m'}{\pc'} $$
        Since $k>0$ it must be the case that $i' < k+1$ and since $I \subseteq (\loweq{\level_{\mathit{auth}} \sqcup \level_{\adv}})$ we get by \fullref{lemma:propagation-ell-equivalence} that
        $\indprop{\indmeta}{c}{\pc}{\level_{\mathit{auth}} \sqcup \level_{\adv}}{\beta_1 \dots \beta_j} \subseteq (\loweq{\level_{\mathit{auth}} \sqcup \level_{\adv}})$.
        Hence, since $m \mathrel{\indmeta} s$ we have that $m \loweq{\level_{\mathit{auth}} \sqcup \level_{\mathit{to}}} s$ and finally from that we obtain that
        $m \mathrel{\indprop{\indmeta}{c}{\pc}{\level_{\mathit{auth}} \sqcup \level_{\adv}}{\beta_1 \dots \beta_j}} s$.
        This enables us to conclude that $m_j \loweq{\level_{\mathit{auth}} \sqcup \level_{\adv}} s_j$ and we can therefore apply the induction hypothesis to
        $$ \configThree{c_j}{m_j}{\pc_j} \lbridgestep{i}{\level_{\mathit{auth}} \sqcup \level_{\adv}}{\EventTiniExit{\eta}{\level_{\mathit{auth}}}{\level_{\mathit{from}}}{\level_{\mathit{to}}}} \configThree{c'}{m'}{\pc'} $$
        and obtain what we need from Case~\ref{bridge-item:2}.
      \end{description}

    \item [If $\alpha \not = \EventTiniExit{}{}{}{}$ and $\pc' \sqsubseteq \level_\adv$: ]
      Suppose we are given $s$ such that $m \loweq{\level_\adv} s$.
      We need to show that there exists run
      $\configThree{d_1; d_2}{s}{\pc} \lbridgestep{k'}{\ell_\adv}{\alpha} \configThree{c'}{s'}{\pc'}$
      and, if $\alpha$ is not an observable declassify event, that $m' \loweq{\level_\adv} s'$.

      By applying \fullref{seq_decomp} we have two cases:
      \begin{description}
      \item[$\alpha$ is produced by $d_1$] Then what we need follows directly from applying the inner induction hypothesis.
      \item[$\alpha$ is produced by $d_2$]
        We then have that

        $$\configThree{d_1}{m}{\pc} \lbridgestep{k_1}{\ell_\adv}{\beta} \configThree{\Stop}{m_1}{\pc_1}$$
        and
        $$\configThree{d_2}{m_1}{\pc_1} \lbridgestep{k_2}{\ell_\adv}{\alpha} \configThree{c'}{m'}{\pc'}$$
        where $\nobs{\beta}{\level_\adv}$ and $k_1 + k_2 = k$.
        Since we end with a low program-counter, $\pc' \sqsubseteq \level_\adv$, it follows from \fullref{lem:obs-event-implies-low-pc} that $\pc_1 \sqsubseteq \level_\adv$.
        Therefore, by applying the inner induction hypothesis on
        $\configThree{d_1}{m}{\pc} \lbridgestep{k_1}{\ell_\adv}{\beta} \configThree{\Stop}{m_1}{\pc_1}$
        we obtain run
        $\configThree{d_1}{s}{\pc} \lbridgestep{k_1}{\ell_\adv}{\beta} \configThree{\Stop}{s_1}{\pc_1}$
        such that
        $m_1 \loweq{\level_\adv} s_1$.
        This further entails that $m_1 \mathrel{\indmeta} s_1$ so by applying the inner induction hypothesis on
        $\configThree{d_2}{m_1}{\pc_1} \lbridgestep{k_2}{\ell_\adv}{\alpha} \configThree{c'}{m'}{\pc'}$
        we obtain that
        $\configThree{d_2}{s_1}{\pc_1} \lbridgestep{k''}{\ell_\adv}{\alpha} \configThree{c'}{s'}{\pc'}$
        and $m' \loweq{\level_\adv} s'$
        which is exactly what we need.
      \end{description}
    \item [If $\alpha = \EventTiniExit{\eta}{\level_{\mathit{auth}}}{\level_{\mathit{from}}}{\level_{\mathit{to}}}$ or $\pc' \not \sqsubseteq \level_\adv$:]
      We are given run
      $\configThree{d_1;d_2}{s}{\pc} \lbridgestep{k'}{\ell_\adv}{\alpha'} \configThree{c''}{s'}{\pc''}$ such that $m \loweq{\level_\adv} s$.
      Then either it must be the case that $\nobs {\alpha} {\level}$ or $\alpha = \EventTiniExit{\eta}{\level_{\mathit{auth}}}{\level_{\mathit{from}}}{\level_{\mathit{to}}}$:
      \begin{description}
      \item[$\nobs {\alpha} {\level_\adv}$: ]
        It then follows that the finals memories must be related
        (since none of the run emit any observable events) and that $\pc'' \not \sqsubseteq \level_\adv$.

      \item [$\alpha = \EventTiniExit{\eta}{\level_{\mathit{auth}}}{\level_{\mathit{from}}}{\level_{\mathit{to}}}$:]
        We then have that $\alpha'$ must also be
        $\EventTiniExit{\eta}{\level_{\mathit{auth}}}{\level_{\mathit{from}}}{\level_{\mathit{to}}}$
        since otherwise it would have to be an unobservable event
        and $c''$ would have to be $\Stop$, which leads to a contradiction since $\declWF{d_1;d_2}$ holds and therefore the run in $s$ cannot ``step over'' the \code{pcdecl} command.
        It then also follows that $m' \loweq{\level_\adv} s'$ since both runs only emit unobservable events up to the $\EventTiniExit{\eta}{\level_{\mathit{auth}}}{\level_{\mathit{from}}}{\level_{\mathit{to}}}$ event
        and the commands $c'$ and $c''$ must be the same.
      \end{description}
    \end{description}
  \item[$c$ is conditional $\IfThenElse{\expr}{c_t}{c_e}$:]
    Suppose
    $\configTwo{\expr}{m} \yields \LabeledValue{\basemeta}{\level_v}$.
    We consider two cases based on whether or not $\level_v \flowsto \level_\adv$.
    \begin{description}
    \item[$\level_v \flowsto \level_\adv$:]
      From \fullref{lem:NI-expressions} we have that for any memory $s$ such that $m \loweq{\level_\adv} s$, it holds that
      $\configTwo{\expr}{s} \yields \LabeledValue{\basemeta}{\level_v}$.
      Hence, any other memory will also be able to step, and it will step to the same branch.
      After stepping to euther $c_t$ or $c_e$ we are done by applying the inner induction hypothesis.
    \item[$\level_v \not \flowsto \level_\adv$:]
      We know that after stepping to one of the branches,
      the program-counter will be $\pc \sqcup \level_v$ for which it holds that $\level_\adv \not \sqsubseteq \pc \sqcup \level_v$.
      Now since we know that $\declWF{c}$ holds, we know that the $\pc$ cannot go down in either branch,
      so we can apply Lemma~\ref{lem:pc-increases-pcdeclfree} to conclude that $\pc' \not \sqsubseteq \level_\adv$.
      Hence we know that $c' = \Stop$ and are given $s$ such that $m \loweq{\level_\adv} s$.
      We need to show
      \begin{multline*}
        \configThree{\IfThenElse{\expr}{c_t}{c_e}}{s}{\pc} \lbridgestep{k+1}{\ell_\adv}{\alpha'} \configThree{c''}{s'}{\pc''}
        \implies
        c' = c'' = \Stop \land m' \loweq{\level_\adv} s' \land \nobs{\alpha'}{\level_\adv}
      \end{multline*}
      which follows directly by applying the inner induction hypothesis to either of the branches that the run may step to.
    \end{description}
  \item[$c$ is $\While{\expr}{c_b}$:]
    Follows from unfolding a single evaluation step and applying the same reasoning as above for $\IfThenElse{\expr}{c_b}{\Skip}$ and sequences.
  \item[$c$ is $\Tini{\eta}{\level_{\mathit{to}}}{\expr_{\mathit{auth}}}{c_b}$:]
    Follows from unfolding a single evaluation step and applying the same reasoning as above for sequences.
  \item [$c$ is $\Eval {e} { \{x_1, \ldots, x_n \} }$:]
    By unfolding a single evaluation step we obtain that $e$ evaluates to a string that can be parsed as a command $c$.
    Now the reasoning is exactly the same as for conditionals $\IfThenElse{}{}{}$ since we are either in a case where $e$ is a ``low'' value and then any other memory will produce the same string
    and otherwise we are stepping to a ``high'' $\pc$ and we can again reason in the same fashion as for conditionals.
  \end{description}
\end{description}

\subsection{Proof of top-level definitions}

\paragraph{Proof of soundness for progress-sensitive NI with declassification and bounded PINI}
We are now in position to prove Theorem~\ref{thm:soundness-bpni}:
Suppose we have an attacker at level $\level$ and a run
\begin{equation*}
  \configThree{\cmd}{\mem}{\pc}
  \stepsmanywith{\events \cdot \alpha}
  \configThree{\cmd'}{\mem'}{\pc'}
\end{equation*}
where $\levelof{\eventMeta} \flowsto \level$.

Suppose
$\tproj{\events}{\level} = \beta_1, \beta_2, \ldots, \beta_j$:
Then it must be the case that
\begin{displaymath}
  \configThree{\cmd}{\mem}{\pc}
  \lbridgestep{i_1}{\level}{\beta_1}
  \configThree{\cmd_1}{\mem_1}{\pc_1}
  \lbridgestep{i_2}{\level}{\beta_2}
  \ldots
  \lbridgestep{i_j}{\level}{\beta_j}
  \configThree{\cmd_j}{\mem_j}{\pc_j}
  \lbridgestep{i}{\level}{\alpha}
  \configThree{\cmd'}{\mem'}{\pc'}
\end{displaymath}

We have three cases to show:
\begin{description}[style=nextline]
\item[$\alpha$ is $\EventDeclassify{x}{\ell_{\mathit{from}}}{\ell_{\mathit{auth}}}{\ell_{\mathit{to}}}$:]
  We need to show 2 conditions:
  \begin{align}
    \label{case:11} \progressknowledge(\cmd, \mem, \events, \level) \supseteq \knowledge(\cmd, \mem, \events, \level) \\
    \label{case:12} \knowledge(\cmd, \mem, \events \cdot \alpha, \level) \supseteq \knowledge(\cmd, \mem, \events, \ell_{\mathit{auth}} \sqcup \ell)
  \end{align}
  \begin{description}
  \item[$\progressknowledge(\cmd, \mem, \events, \level) \supseteq \knowledge(\cmd, \mem, \events, \level)$: ]
    To show Condition~\ref{case:11}, suppose $s \in \knowledge(\cmd, \mem, \events, \level)$.
    By unfolding the knowledge definition, this entails that
    $\configThree{c}{s}{\pc} \stepsmanywith{\events'} \configThree{c''}{s''}{\pc''}$
    where $\tproj{\events'}{\level} = \tproj{\events}{\level}$ and $m \loweq{\level} s$.
    So using \fullref{lem:runs-to-sync-bridge} we have that
    \begin{displaymath}
      \syncConfig{c}{m}{s}{\pc} \syncbridge{\ell}{\beta_1, \beta_2, \ldots, \beta_j} \syncConfig{c_j}{m_j}{s_j}{\pc_j}
    \end{displaymath}
    Which we can use to conclude that $m \indprop{\loweq{\level}}{c}{\pc}{\level}{\beta_1, \dots \beta_j} s$,
    so by \fullref{lemma:propagation-ell-equivalence} we have that $m_j \loweq{\level} s_j$.

    We need to show that $s \in \progressknowledge(\cmd, \mem, \events, \level)$, which now amounts to showing
    \begin{displaymath}
      \configThree{c_j}{s_j}{\pc_j} \lbridgestep{k}{\level}{\EventDeclassify{x}{\ell_{\mathit{from}}}{\ell_{\mathit{auth}}}{\ell_{\mathit{to}}}} \configThree{c'}{s'}{\pc'}
    \end{displaymath}
    Now, it must be the case that $\pc' \flowsto \level$ (since ${\EventDeclassify{x}{\ell_{\mathit{from}}}{\ell_{\mathit{auth}}}{\ell_{\mathit{to}}}}$ is observable),
    so this follows directly from Case~\ref{bridge-item:3} of \fullref{lem:NI-bridge}.

  \item[$\knowledge(\cmd, \mem, \events \cdot \alpha, \level) \supseteq \knowledge(\cmd, \mem, \events, \ell_{\mathit{auth}} \sqcup \ell)$:]
    To show Condition~\ref{case:12}, suppose $s \in \knowledge(\cmd, \mem, \events, \ell_{\mathit{auth}} \sqcup \ell)$.
    By unfolding the knowledge definition, this entails that
    $\configThree{c}{s}{\pc} \stepsmanywith{\events'} \configThree{c''}{s'}{\pc''}$
    where $\tproj{\events'}{\ell_{\mathit{auth}} \sqcup \ell} = \tproj{\events}{\level}$ and $m \loweq{\ell_{\mathit{auth}} \sqcup \ell} s$.
    So using \fullref{lem:runs-to-sync-bridge} we have that
    $$\syncConfig{c}{m}{s}{\pc} \syncbridge{\ell_{\mathit{auth}} \sqcup \ell}{\beta_1, \beta_2, \ldots, \beta_j} \syncConfig{c_j}{m_j}{s_j}{\pc_j}$$
    where
    $\configThree{c_j}{m_j}{\pc_j} \lbridgestep{i}{\level}{\EventDeclassify{x}{\ell_{\mathit{from}}}{\ell_{\mathit{auth}}}{\ell_{\mathit{to}}}} \configThree{c'}{m'}{\pc'} $

    We then need to show that
    $s \in \knowledge(\cmd, \mem, \events \cdot \alpha, \level)$,
    which now amounts to showing
    $$\configThree{c_j}{s_j}{\pc_j} \lbridgestep{i'}{\level}{\EventDeclassify{x}{\ell_{\mathit{from}}}{\ell_{\mathit{auth}}}{\ell_{\mathit{to}}}} \configThree{c'}{s'}{\pc'}$$

    We have that
    $m \loweq{\ell_{\mathit{auth}} \sqcup \ell} s$
    so from the synchronized bridge above it follows that
    $m_j \indprop{(\loweq{\ell_{\mathit{auth}} \sqcup \ell})}{c}{\pc}{\ell_{\mathit{auth}} \sqcup \ell}{\beta_1 \dots \beta_j} s_j$.
    Using \fullref{lemma:propagation-ell-equivalence} we can also conclude that
    $\indprop{(\loweq{\ell_{\mathit{auth}} \sqcup \ell})}{c}{\pc}{\ell_{\mathit{auth}} \sqcup \ell}{\beta_1 \dots \beta_j} \subseteq (\loweq{{\ell_{\mathit{auth}} \sqcup \ell}})$, so it must be the case that $m_j \loweq{{\ell_{\mathit{auth}} \sqcup \ell}} s_j$.
    Now we can obtain exactly what we need from Case~\ref{bridge-item:1} of \fullref{lem:NI-bridge}.
  \end{description}
\item[$\alpha$ is $\EventTiniExit{\eta}{.}{\ell_{\mathit{from}}}{.}$:]
  We need to show that two conditions:
  \begin{align}
    \label{case:21}
    \knowledge(\cmd, \mem, \events \cdot \alpha, \level) \supseteq \progressknowledge(\cmd, \mem, \events, \ell) \\
    \label{case:22}
    \progressknowledge(\cmd, \mem, \events, \ell) \supseteq \knowledge(\cmd, \mem, \events, \ell_{\mathit{auth}} \sqcup \ell)
  \end{align}
  \begin{description}
  \item[$\knowledge(\cmd, \mem, \events \cdot \alpha, \level) \supseteq \progressknowledge(\cmd, \mem, \events, \ell)$:]
    To show Condition~\eqref{case:21}, suppose $s \in \progressknowledge(\cmd, \mem, \events, \ell)$.
    By unfolding the definition of progress knowledge we therefore have that
    $\configThree{c}{s}{\pc} \stepsmanywith{\events'} \configThree{c''}{s'}{\pc''}$
    where $\tproj{\events'}{\level} = \tproj{\events}{\level} \cdot \alpha'$ for some $\alpha'$ and that $m \loweq{\level} s$.
    We need to show that $s \in \knowledge(\cmd, \mem, \events \cdot \alpha, \level)$, which amounts to showing that
    $\alpha = \alpha'$.
    Now since we know that $\alpha$ is a public event, we know that $\pc' \sqsubseteq \ell$ and hence we obtain what we need from Case~\ref{bridge-item:4} of applying \fullref{lem:NI-bridge}.
  \item[$\progressknowledge(\cmd, \mem, \events, \ell) \supseteq \knowledge(\cmd, \mem, \events, \ell_{\mathit{auth}} \sqcup \ell)$:]
    To show Condition~\eqref{case:22}, suppose $s \in \knowledge(\cmd, \mem, \events, \level \sqcup \ell_{\mathit{auth}} \sqcup \ell)$.
    By unfolding the definition of knowledge this entails that
    $\configThree{c}{s}{\pc} \stepsmanywith{\events'} \configThree{c''}{s'}{\pc''}$
    where $\tproj{\events'}{\level_{\mathit{auth}} \sqcup \level} = \tproj{\events}{\level_{\mathit{auth}} \sqcup \level}$ and $m \loweq{\level_{\mathit{auth}} \sqcup \level} s$.
    So using \fullref{lem:runs-to-sync-bridge} we have that
    $$\syncConfig{c}{m}{s}{\pc} \syncbridge{\ell_{\mathit{auth}} \sqcup \ell}{\beta_1, \beta_2, \dots, \beta_j} \syncConfig{c_j}{m_j}{s_j}{\pc_j}$$
    where
    $\configThree{c_j}{m_j}{\pc_j} \lbridgestep{i}{\EventTiniExit{\eta}{\ell_{\mathit{auth}}}{\level}{\ell_{\mathit{to}}}}{\level} \configThree{c'}{m'}{\pc'} $
    We need to show that
    $s \in \progressknowledge(\cmd, \mem, \events, \level)$
    which now amounts to showing that there exists $\alpha'$ such that
    $$\configThree{c_j}{s_j}{\pc_j} \lbridgestep{i'}{\level}{\alpha'} \configThree{c''}{s'}{\pc''}$$
    
    We have that
    $m \loweq{\level_{\mathit{auth}} \sqcup \level} s$ so from the synchronized bridge above it follows that
    $m_j \indprop{(\loweq{\level_{\mathit{auth}} \sqcup \level})}{c}{\pc}{\level_{\mathit{auth}} \sqcup \level}{\beta_1, \dots, \beta_j} s_j$.
    Using \fullref{lemma:propagation-ell-equivalence} we can also conclude that
    $\indprop{(\loweq{\level_{\mathit{auth}} \sqcup \level})}{c}{\pc}{\level_{\mathit{auth}} \sqcup \level}{\beta_1 \dots \beta_j} \subseteq (\loweq{\level_{\mathit{auth}} \sqcup \level})$, so it must be the case that $m_j \loweq{\level_{\mathit{auth}} \sqcup \level} s_j$.
    Again, we are now in a position to obtain exactly what we need from Case~\ref{bridge-item:2} of applying \fullref{lem:NI-bridge}.
  \end{description}
\item[Otherwise:]
  We need to show
  \begin{align*}
    \knowledge(\cmd, \mem, \events \cdot \alpha, \level)
    \supseteq \knowledge(\cmd, \mem, \events, \level)
  \end{align*}
  So suppose $s \in \knowledge(\cmd, \mem, \events, \level)$.
  This entails that there exists run such that
  $\configThree{c}{s}{\pc} \stepsmanywith{\events'} \configThree{c''}{s''}{\pc''}$
  where $\tproj{\events'}{\level} = \tproj{\events}{\level} = \beta_1, \beta_2,
  \ldots , \beta_j$.

  Again, using \fullref{lem:runs-to-sync-bridge} we can conclude that
  $$\syncConfig{c}{m}{s}{\pc} \syncbridge{\ell}{\beta_1, \beta_2, \ldots, \beta_j} \syncConfig{c_j}{m_j}{s_j}{\pc_j}$$
  such that
  \begin{displaymath}
    \configThree{c_j}{m_j}{\pc_j} \lbridgestep{i}{\level}{\alpha} \configThree{c'}{m'}{\pc'}
  \end{displaymath}
  which we can use to conclude that $m \indprop{\loweq{\level}}{c}{\pc}{\level}{\beta_1, \dots \beta_j} s$,
  so by \fullref{lemma:propagation-ell-equivalence} we have that $m_j \loweq{\level} s_j$.

  We need to show that
  $s \in \knowledge(\cmd, \mem, \events \cdot \alpha, \level)$ which amounts to showing that
  $$\configThree{c_j}{s_j}{\pc_j} \lbridgestep{m'}{\level}{\alpha} \configThree{c'}{s'}{\pc'}$$
  Now since we know that $\obs{\alpha}{\level}$, we know that $\pc' \sqsubseteq \level$.
  Hence we can conclude what we need from Case~\ref{bridge-item:3} of applying \fullref{lem:NI-bridge}.
\end{description}

%%% Local Variables:
%%% mode: latex
%%% TeX-master: "csf2020"
%%% End:
}

\end{document}